\documentclass[11pt]{article}
\usepackage{amsmath, amsthm, amssymb, amsfonts}
\usepackage[hmargin={1.0in,1.0in},vmargin={1.0in,1.0in}]{geometry}
\usepackage{setspace}
\usepackage{url}
\usepackage{marvosym}
\usepackage{mathtools}
\usepackage{hyperref}

\theoremstyle{plain}
\newtheorem{theorem}{Theorem}
\newtheorem{lemma}{Lemma}
\newtheorem{proposition}{Proposition}
\newtheorem{assumption}{Assumption}
\newtheorem{remark}{Remark}
\newtheorem{definition}{Definition}
\newtheorem{corollary}{Corollary}
\newtheorem{example}{Example}

\begin{document}

\title{Implicit transaction costs and the fundamental theorems of asset pricing  }
\author{Erindi Allaj \\ University of Rome Tor Vergata, Via Columbia 2, 00133 Rome, Italy \\erindi.allaj@uniroma2.it}
\maketitle
\begin{abstract}
This paper studies arbitrage pricing theory in financial markets with implicit transaction costs.  We extend the existing theory to include the more realistic possibility that the price at which the investors trade is dependent on the traded volume. The investors in the market always buy at the ask and sell at the bid price. Implicit transaction costs are composed of two terms, one is able to capture the bid-ask spread, and the second the price impact. Moreover, a new definition of a self-financing portfolio is obtained. The self-financing condition suggests that continuous trading is possible, but is restricted to predictable trading strategies  having c\'adl\'ag (right-continuous with left limits) and c\'agl\'ad (left-continuous with right limits) paths of bounded quadratic variation and of finitely many jumps. That is, c\'adl\'ag and c\'agl\'ad predictable trading strategies of infinite variation, with finitely many jumps and of finite quadratic variation are allowed in our setting. Restricting ourselves to c\'agl\'ad predictable trading strategies, we show that the existence of an equivalent probability measure is equivalent to the absence of arbitrage opportunities, so that the first fundamental theorem of asset pricing (FFTAP) holds. It is also shown that the use of continuous and bounded variation trading strategies can improve the efficiency of hedging in a market with implicit transaction costs. To better understand how to apply the theory proposed we provide an example of an implicit transaction cost economy that is linear and non-linear in the order size. 
\end{abstract}
\section{Introduction}
The subject of this paper is the study of the arbitrage pricing theory in continuous time markets with implicit transaction costs. These implicit costs are measured by comparing marginal prices with the marginal benchmark price, given by the marginal mid-price, and by the comovement of the ask and the bid price with the trading strategy. To each trade corresponds a different price, the ask and the bid price, depending on whether the trade is a buy or a sell. The trade prices are assumed to depend on the order size of the trade. 

The standard arbitrage pricing theory assumes, among other things, that every asset in the market can be traded without any transaction costs. A large body of  theoretical research in economics and finance has relaxed the assumption of no transaction costs in continuous time setting. The impact of transaction costs on the investment decision making  has been analyzed on several papers in the economic literature such as Amihud \& Mendelson (1986), Constantinides (1986), Davis \& Norman (1990), Dumas \& Luciano (1991),  Gerhold \textit{et al.} (2013), and many other papers.  

Several papers in the finance literature also analyze the effects of the introduction of (implicit) transaction costs on the standard arbitrage pricing theory. For example, Guasoni (2006) introduces a simple criterion for the absence of arbitrage opportunities under proportional transaction costs and under some additional assumptions regarding the return process. Several other valuable papers try to prove different versions of the first fundamental theorem of asset pricing (FFTAP) under transaction costs. Guasoni \textit{et al.} (2010) prove a version of the FFTAP for continuous time market models with small proportional transaction costs. The paper by Denis \textit{et al.} (2012) proves also a version of the FFTAP with transaction costs where the bid and the ask prices are assumed to be locally bounded c\'adl\'ag (right-continuous with left limits) processes. Other papers including also the discrete time case are given, for example, by Cherny (2007), Denis \& Kabanov (2012), Jouini \& Kallal (1995), and L\'epinette \& Tran (2017). In particular, under the discrete time setting, the last paper outlines some no arbitrage criteria for market models with general non-proportional transaction costs. 

There is also a considerable literature dealing with the problem of the derivative contract hedging in the presence of transaction costs. In a market with proportional transaction costs, Leland (1995) introduces a sophisticated method in order to hedge the European call option on a discrete time scale. Kabanov \& Safarian (1997) show that the value of the replicating portfolio converges to the payoff of the European call option for arbitrary small transaction costs. For other related literature on the hedging problem, see also Bensaid \textit{et al.} (1992), Hodges \& Neuberger (1989), and Soner \textit{et al.} (1995).  

In reviewing the finance literature studying the arbitrage pricing theory with transaction costs and in a continuous time setting, three things are noted. First, in all of the studies, admissible trading strategies are restricted to trading strategies that are of bounded variation. Second, the FFTAP generally states the equivalence between no arbitrage opportunities and the existence of a consistent price system. Lastly, the ask and the bid price are supposed to not depend on the traded volume. 

To overcome these limits, we propose an implicit transaction cost economy made up with one risky and one riskless asset, where the investors buy and sell the risky asset at the ask and the bid price.  We suppose these prices depend not only on time, but also on the traded volume, i.e. $A(t,y)$ and $B(t,y)$, where $y>0$ gives a buy order and $y<0$ a sell order. The influence of the traded volume on the trade price is evidenced in the literature by different authors, for example, Almgren \& Chriss (2001), Bank \& Baum (2004), Barclay \& Warner (1993), Bertsimas \& Lo (1998), Blais \& Protter (2010), Engle \& Patton (2004), Guasoni \& R\'asonyi (2015), Hasbrouck (1991), and Schied \& Sch\"{o}neborn (2009). A common feature of all these studies is that traded volume has an impact on the trade price. The only difference being the size of the traded volume. 

The prices $A(t,y)$ and $B(t,y)$ are supposed to be equal to $M(t,y)+\frac{1}{2}P(t,y)$ and $M(t,y)-\frac{1}{2}P(t,y)$. $M(t,y)$ and $P(t,y)$ are assumed to be non-negative and $C^2$ in the traded volume and $M(t,0)$, $P(t,0)$ c\'adl\'ag locally bounded semimartingales. $M(t,0)$ and $P(t,0)$ give the marginal mid-price and the marginal bid-ask spread corresponding to an infinitesimal purchase or sale. The admissible trading strategies are allowed to be c\'adl\'ag predictable and c\'agl\'ad adapted with bounded quadratic variation and finitely many jumps.  The self-financing condition is composed of the usual standard self-financing condition, bid-ask spread part,  and the price impact measured by the changes in the price that are created when trading on a given asset. 

Our approach is strickly related to the works of Bank \& Baum (2004), \c{C}etin \textit{et al.} (2004), 
Jarrow (1992) and Jarrow (1994). A common feature of these works is that the trades of the large traders have an impact on prices. Here, as in \c{C}etin \textit{et al.} (2004), prices are assumed to depend only on the investors' current trade. 

The no arbitrage (NA) property is still true in our framework. By limiting ourselves to c\'agl\'ad adapted trading strategies, we also prove the no free lunch with vanishing risk (NFLVR) property. The proof of the NA property proceeds more or less along the same lines as in the standard proof of the NA property. On the other hand, the proof of the FFTAP starts by first proving that the set of all admissible portfolios is bounded in $L^0$. This also permits us to conclude that  the set composed of the continuous part of the quadratic variations of all admissible trading strategies is bounded in $L^0$. Next, supposing the NFLVR holds, we prove a compactness result which shows that every sequence of bounded quadratic variation, admissible trading strategies has a sequence of convex combinations which converges a.s. to a bounded quadratic variation trading strategy. More importantly, we show that any sequence of bounded quadratic variation, admissible trading strategies has a subsequence which converges a.s. to a bounded quadratic variation trading strategy. Using this important result, we then prove that the set of all admissible portfolios that can be dominated by an admissible portfolio is Fatou-closed. As a corollary of the above proofs, we show that by restricting the set of admissible portfolios to trading strategies that are continuous and of bounded variation, the set becomes a convex cone and is  not only Fatou-closed but also weak*-closed. Finally, equipped with these results we give the proof of the FFTAP. 

It is worth mentioning that the recent paper by L\'epinette \& Tran (2017) proves that, under some conditions on the structure of the transaction costs, the FFTAP stated as in the standard arbitrage pricing theory is still valid under a discrete time setting and under general non-linear transaction costs. The equivalent probability measure rules out arbitrage opportunities  that hold also in financial models with linear transaction costs. From this point of view, the models are equivalent, in the sense that a risk-neutral probability measure in our setting corresponds to zero arbitrage opportunities in the model with no bid-ask spread part (linear transaction costs part) and price impact part, and zero arbitrage opportunities in the model with only the bid-ask spread part (linear transaction costs part). 

Although the proposed economy satisfies the FFTAP, it is incomplete. Even though the market is incomplete (in an $L^2$-sense), we show that the use of continuous and bounded variation trading strategies is the best choice when dealing with the problem of hedging. The use of these strategies allows investors to remove the price impact from portfolio's value. This result suggests that investors can split a large trade into infinitesimally small trades in order to avoid the price impact. To prove this, we use arguments similar  to those used by \c{C}etin \textit{et al.} (2004) in the proof of the second fundamental theorem of asset pricing (SFTAP).  The authors develop a mathematical model including liquidity risk into the standard arbitrage pricing theory. The basic idea of the proof is to assume that there exists a second type of economy, called the standard economy, without implicit transaction costs in addition to the implicit transaction cost economy. 

Taken together, these results make three principal contributions to arbitrage pricing theory. First, they 
show that the FFTAP is valid even in the presence of implicit transaction costs. Indeed, NFLVR is equivalent to the existence of an equivalent probability measure. When this probability measure is unique, the results show that in markets with implicit transaction costs trading with continuous trading strategies of bounded variation can improve hedging. Last but not least, they are proved under the realistic assumption that the trade price depends on the traded volume. 

Another important result shown in the paper is that continuous trading is not only limited to bounded variation trading strategies. In particular, infinite variation trading strategies with bounded quadratic variation and finitely many jumps are allowed in our setting. A notable example is given by the replicating trading strategy of the European call/put option. 

Overall, we believe that these findings will improve the understanding of the effects of the transaction costs on the arbitrage pricing theory. 

The structure of the paper is as follows. Sec. \ref{the model} presents the model, describing 
its main assumptions. Sec. \ref{sec:3} derives the self-financing portfolio for predictable c\'adl\'ag  and c\'agl\'ad trading strategies having bounded quadratic variation and finitely many jumps. Sec. \ref{na} is dedicated to the proof of the FFTAP. Sec. \ref{sec2} studies market completeness. Sec. \ref{bs} provides an example of a linear and a non-linear implicit transaction cost economy by extending the obtained results to the case of the Black-Scholes (BS) model.  Sec. \ref{conc} concludes the paper.   
\section{The market model\label{the model}}
For a fixed time trading horizon $[0, T]$, consider the filtered probability space $(\Omega, \mathcal{F}, \textbf{F}, \mathbb{P})$. The filtration $\textbf{F}=(\mathcal{F}_{t})_{(0 \leq t \leq T)}$ satisfies the usual conditions of completeness and right-continuity and $\mathbb{P}$ denotes the reference probability measure. The sigma algebra $\mathcal{F}$ is generated by $\cup_{t \in [0, T]} \mathcal{F}_{t}$, and  $\mathcal{F}_{0}$ is trivial, i.e. $\mathcal{F}_{0}=\left\{\emptyset, \Omega\right\}$. 

The reference economy is composed of two assets, one riskless and one risky. The riskless asset plays the role of the numeraire, and for simplicity, it is assumed to be constant, i.e. has a zero rate of return. The risky asset is a stock. The price at which shares of the stock can be bought or sold is different. 
The representative investor builds a portfolio by combining an investment mix of stock and riskless asset. 

\subsection{Trading strategies\label{trading}}
\begin{definition}\label{tradin}
A trading strategy, or portfolio, is given by $Z_t=(Z^0_{t}, Z^1_{t})_{t \in [0,T]}$, where $Z^0_{t}$ and $Z^1_{t}$ denote the  number of units held at time $t$ of the riskless asset and  the risky asset. The trading strategies can take one of the following forms
\begin{itemize}
\item[1.] $Z^1$ is a c\'adl\'ag predictable process, with finitely many jumps, and of bounded quadratic variation on $[0,T]$, and $Z^0$ an optional process on $[0,T]$;
\item[2.] $Z^1$ is a c\'agl\'ad adapted process, with finitely many jumps, and of bounded quadratic variation on $[0,T]$, and $Z^0$ a l\'adl\'ag (right-limited and left-limited) adapted process on $[0,T]$.
\end{itemize}
\end{definition}
The following definition specifies the concept of the quadratic variation of a real-valued stochastic process $X$. 
\begin{definition}\label{quadratic}
Let $\sigma_{n}:0=\tau_{0}^{n}\leq\tau_{1}^{n}\leq...\leq\tau_{i_{n}}^{n}=T$ be a sequence of random partitions of the interval $[0,T]$ tending to identity (see Protter (2004)), where $\tau_{i}^{n}$'s are stopping times. We then say that $X$ has bounded quadratic variation $[X,X]_{T}$ on a given interval $[0,T]$ if for every sequence $\{\sigma_{n}\}$  the quantity
\begin{equation}
[X,X]_{T}=\lim_{n\rightarrow\infty}\sum_{i\geq1}|X_{\tau_{i}^{n}}-X_{\tau_{i-1}^{n}}|^{2}
\end{equation}
exists and is finite a.s..
\end{definition}
\subsection{Transaction costs}  
Trading in the stock, that is, buying or selling the stock, incurs transaction costs. According to Keim \& Madhavan (1995), these costs can be divided in explicit and implicit. Examples of explicit costs include brokerage commissions, administrative costs, and transaction securities taxes. On the other hand, the measurement of the implicit transaction costs is usually concerned with comparing the transaction price, namely, the price that the investor pays or receives for the stock, and the price that would prevail without the trade happening.  While explicit costs are easy to track in practice, the measurement of implicit costs would require a benchmark price against which to compare the transaction and the non-transaction price. 

This paper ignores the explicit costs and focus only on the implicit costs. In this subsection, we look at how these costs can be explicitly measured. The next section shows how these costs can be included in a self-financing portfolio.

Let $Z^1_{\tau_{i-1}}$ be the amount of stock held by the investor on the interval $(\tau_{i-1}, \tau_{i}]$, and suppose that at time $\tau_{i}$ the investor adjusts his holdings from $Z^1_{\tau_{i-1}}$ to $Z^1_{\tau_{i}}$ the amount to be held over the interval $(\tau_{i}, \tau_{i+1}]$. Further suppose the price the investor could sell and buy the stock is $B(\tau_{i})$ and $A(\tau_{i})$, respectively the bid and the ask price, with ask price being greater or equal than the bid price. The size of the transaction amounts to  $Z^1_{\tau_{i}}-Z^1_{\tau_{i-1}}$, while the monetary value to $(Z^1_{\tau_{i}}-Z^1_{\tau_{i-1}})A(\tau_{i})$ or $(Z^1_{\tau_{i}}-Z^1_{\tau_{i-1}})B(\tau_{i})$. 

By convention, a  positive sign of $Z^1_{t_{i}}-Z^1_{\tau_{i-1}}$ indicates a buy, a negative sign a sale, and $Z^1_{\tau_{i}}-Z^1_{\tau_{i-1}} = 0$ indicates no trading in the risky asset at time $\tau_{i}$.
\begin{assumption}\label{mid}
The prices $A$ and $B$ are assumed to depend on time and on the order size in the following form,  i.e. $A=A(t,y)=M(t,y)+\frac{1}{2}P(t,y)$ and $B=B(t,y)=M(t,y)-\frac{1}{2}P(t,y)$, where $M(t,y)=\frac{1}{2}[A(t,y)+B(t,y)]$ and $P(t,y)=A(t,y)-B(t,y)$. The non-negative quantities $M(t,y)$ and $P(t,y)$ are adapted to the filtration  $\textbf{F}$, $M(t,0)$ and $P(t,0)$ are c\'adl\'ag locally bounded semimartingales adapted to $\textbf{F}$ with decomposition $X= N +H$, where $N$ is an $\textbf{F}$-local martingale and $H$ a process of bounded variation. Moreover, $M(t,y)$ and $P(t,y)$ are $C^2$ in $y \in\mathbb{R}$, and $M'(t,y)$, $P'(t,y)$ together with $M''(t,y)$, $P''(t,y)$ are c\'adl\'ag locally bounded in $t$, where $\frac{\partial X(t,y)}{\partial y}=X'(t,y)$ and $\frac{\partial^2 X(t,y)}{\partial^2 y}=X''(t,y)$.  
\end{assumption}
To justify the functional forms choosen for $A$ and $B$ we exploit the properties of $M$ and $P$ to write the prices $A$ and $B$ as
\begin{eqnarray}
A(t,y)&=&[M(t,0)+\frac{1}{2}P(t,0)] + [M'(t,0)+\frac{1}{2}P'(t,0)]y\nonumber\\&+&o^{M}(t,y)+\frac{1}{2}o^{P}(t,y)\\
B(t,y)&=&[M(t,0)-\frac{1}{2}P(t,0)] + [M'(t,0)-\frac{1}{2}P'(t,0)]y\nonumber\\&+&o^{M}(t,y)-\frac{1}{2}o^{P}(t,y)
\end{eqnarray}
where  $A(t, 0)$, $B(t, 0)$, $M(t, 0)$, and $P(t, 0)$ are the marginal ask, bid, mid-price, and bid-ask spread. Taylor formulas show that the ask and the bid price are the sum of the marginal ask and the marginal bid price corresponding to an economy without price impact, the sensitivity of the ask and the bid price on the trade size, and the errors in approximating these relationships.  

Explicit transaction costs are easy quantifiable in practice since they involve direct monetary payments. Implicit transaction costs by contrast are more difficult to measure because they do not lead to a physical exchange of money. Commonly in finance literature and by most professional traders, they are determined as the product between the trade size and the difference between the transaction prices and a benchmark price. A widely used benchmark is given by the quotation mid-point. 

In our setting, an obvious measure of the implicit transaction costs can be obtained as follows. 
\begin{assumption}\label{ass2}
Given a transaction of monetary value $(Z^1_{\tau_{i}}-Z^1_{\tau_{i-1}})A(\tau_{i}, Z^1_{\tau_{i}}-Z^1_{\tau_{i-1}})$ $(or \quad (Z^1_{\tau_{i}}-Z^1_{\tau_{i-1}})B(\tau_{i}, Z^1_{\tau_{i}}-Z^1_{\tau_{i-1}}))$, the implicit transaction costs are computed as $(Z^1_{\tau_{i}}-Z^1_{\tau_{i-1}})[A(\tau_{i}, 0)-M(\tau_{i}, 0)]+(Z^1_{\tau_{i}}-Z^1_{\tau_{i-1}})[A(\tau_{i}, Z^1_{\tau_{i}}-Z^1_{\tau_{i-1}})-A(\tau_{i}, 0)]$ $(or \quad -(Z^1_{\tau_{i}}-Z^1_{\tau_{i-1}})[M(\tau_{i}, 0)-B(\tau_{i}, 0)] +(Z^1_{\tau_{i}}-Z^1_{\tau_{i-1}})[B(\tau_{i}, Z^1_{\tau_{i}}-Z^1_{\tau_{i-1}})-B(\tau_{i}, 0)])$.
\end{assumption}
Note that $(Z^1_{\tau_{i}}-Z^1_{\tau_{i-1}})[A(\tau_{i}, 0)-M(\tau_{i}, 0)]$ or $-(Z^1_{\tau_{i}}-Z^1_{\tau_{i-1}})[M(\tau_{i}, 0)-B(\tau_{i}, 0)]$ is able to capture the implicit transaction costs for orders that are infinitesimally small or the implicit transaction costs in an economy without price impact. In fact, $A(\tau_{i}, 0)-M(\tau_{i}, 0)$ and $M(\tau_{i}, 0)-B(\tau_{i}, 0)$ are equal to $\frac{1}{2}P(\tau_{i},0)$ which corresponds to the classical one-half bid-ask spread used in an implicit transaction cost economy without price impact.

By further manipulations it is possible to show that the implicit transaction costs can be written as $(Z^1_{\tau_{i}}-Z^1_{\tau_{i-1}})[A(\tau_{i}, Z^1_{\tau_{i}}-Z^1_{\tau_{i-1}}) -M(\tau_i,0)]$ and $-(Z^1_{\tau_{i}}-Z^1_{\tau_{i-1}})[M(\tau_{i},0) -B(\tau_i, Z^1_{\tau_{i}}-Z^1_{\tau_{i-1}})]$ which show that implicit transaction costs are measured simply as the product between the trade size and the difference between the transaction prices and the benchmark price given by the marginal mid-price. 

Assumption \ref{mid} emphasizes the important role of the order size on the transaction price.  
The standard arbitrage pricing theory is based on the assumption that the order size does not influence   the trade price. While mathematically convenient, in practice the effect of the traded volume on prices cannot be neglected. Indeed, this is what happens in quote-driven markets where a market maker or more than one market maker post their bids and  ask prices based on the amount of a given security. Examples of these markets include the London SEAQ system and NASDAQ.  

The dependence of the ask and the bid price on the traded volume is also present in an order driven market, like the NYSE market. Blume \& Goldstein (1997) provides an excellent introduction to this market. Many empirical works conclude that trade size matters in determining the ask and the bid price, but  only up to a certain interval of the traded volume. For example,  Engle \& Patton (2004) analyze the quote price dynamics of 100 NYSE stocks through an error-correction model, and find that a trade with volume between $1,000$ and $10,000$ shares has a significant influence on the ask and the bid price. On the contrary, Hasbrouck (1991)  using a vector autoregressive model for a sample of NYSE stocks  concludes that large trades increase the bid-ask spread and the mid-price more than the small trades.  

For a good overview of the market microstructure foundations see Hautsch (2012).
\begin{example}\label{ex1}
Consider a setting in which the order size of the investor is  $\Delta Z^1_{\tau_{i}}=Z^1_{\tau_{i}}-Z^1_{\tau_{i-1}}$ with monetary value given by $\Delta Z^1_{\tau_{i}}A(\tau_{i},\Delta Z^1_{\tau_{i}})$, when the investor buys, and $\Delta Z^1_{\tau_{i}}B(\tau_{i},\Delta Z^1_{\tau_{i}})$ when he sells. In this case, $A(\tau_{i},\Delta Z^1_{\tau_{i}})$ and $B(\tau_{i},\Delta Z^1_{\tau_{i}})$ give the price for the purchase (sale) of $\Delta Z^1_{\tau_{i}}$ units of the stock. It also follows that the implicit transaction costs amount to $\frac{1}{2}\Delta Z^1_{\tau_{i}}P(\tau_{i},0)+\Delta Z^1_{\tau_{i}}[A(\tau_i,\Delta Z^1_{\tau_{i}})-A(\tau_i,0)]$ and $-\frac{1}{2}\Delta Z^1_{\tau_{i}}P(\tau_{i},0)+\Delta Z^1_{\tau_{i}}[B(\tau_i, \Delta Z^1_{\tau_{i}})-B(\tau_i,0)]$.
\end{example}
\section{Portfolio dynamics}\label{sec:3}  
\subsection{Discrete case}
The portfolio value process at time $\tau_{i}$ is given by
\begin{equation}\label{eq1}
V_{\tau_{i}}^{Z^1} = Z^0_{\tau_{i}} + Z^1_{\tau_{i}} M(\tau_{i},0)
\end{equation}
Recalling what $Z^1$ and $Z^0$ mean, the last equation says that the value of the portfolio has to be evaluated using the marginal mid-price. 

We want to study self-financing portfolios, i.e. portfolios without exogenous infusion or withdrawal of money. The notion of self-financing portfolio becomes more apparent with  simple predictable trading strategies. We choose to work in this paragraph with c\'agl\'ad trading strategies  of the form $Z^1_{u}=Z^1_{0}\textbf{1}_{\{0\}}(u)+\sum_{i=1}^{n} Z^1_{\tau_{i-1}}\textbf{1}_{(\tau_{i-1}, \tau_{i}]}(u)$, where $0=\tau_{0} < \tau_{1}<...<\tau_{n}= T$ is a finite sequence of stopping times and $Z^1_{\tau_{i-1}}$ a $\mathcal{F}_{\tau_{i-1}}$-measurable random variable with $|Z^1_{\tau_{i-1}}| < \infty$. We also include in this trading strategy the  amount of the stock to be held after time $T$, $Z^1_{\tau_{n}}=Z^1_{T}\textbf{1}_{(T, \tau_{n+1}]}(u)$ with $\tau_{n+1}$ finite and $Z^1_{T}$ bounded and $\mathcal{F}_{T}$-measurable. In finance literature it is  often supposed that $Z^1_{0}=0$ and $Z^1_{T}=0$. That is the investor has zero initial holdings in the stock and zero holdings after the trading horizon $T$. The derivation of the self-financing portfolio in the c\'adl\'ag case works in the same way. 
\begin{definition}
A portfolio $Z$ is said to be self-financing if for each $i=1,2,...,n$
\begin{equation}\label{reb}
\Delta Z^0_{\tau_{i}} = -M(\tau_{i},\Delta Z^1_{\tau_{i}})\Delta Z^1_{\tau_{i}}-\frac{P(\tau_{i},\Delta Z^1_{\tau_{i}})}{2}\mbox{sgn}(\Delta Z^1_{\tau_{i}})\Delta Z^1_{\tau_{i}}  
\end{equation}
\end{definition}
where $\Delta Z^0_{\tau_{i}}=Z^0_{\tau_{i}}-Z^0_{\tau_{i-1}}$ and $sgn(X)=1$ for $X>0$, $0$ for $X=0$ and $-1$ for $X <0$. Eq. (\ref{reb}) can be easily obtained by considering two distinct cases, $\Delta Z^1_{\tau_{i}}\geq 0$ and $\Delta Z^1_{\tau_{i}}<0$. If $\Delta Z^1_{\tau_{i}}\geq 0$, the amount of the riskless asset (cash) exchanged by the investor is
\begin{equation}
\Delta Z^0_{\tau_{i}} = -(M(\tau_{i},\Delta Z^1_{\tau_{i}}) + \frac{P(\tau_{i},\Delta Z^1_{\tau_{i}})}{2})\Delta Z^1_{\tau_{i}}
\end{equation}
while
\begin{equation}
\Delta Z^0_{\tau_{i}}=-(M(\tau_{i},\Delta Z^1_{\tau_{i}}) - \frac{P(\tau_{i},\Delta Z^1_{\tau_{i}})}{2})\Delta Z^1_{\tau_{i}} 
\end{equation}
for $\Delta Z^1_{\tau_{i}}<0$. 

That is,  the investor pays the ask price when he buys and the bid price when he sells. The self-financing equation can thus be stated as follows
\begin{equation}\label{self}
\Delta Z^0_{\tau_{i}} = -M(\tau_{i},\Delta Z^1_{\tau_{i}})\Delta Z^1_{\tau_{i}}-\frac{1}{2}P(\tau_{i},\Delta Z^1_{\tau_{i}})\mbox{sgn}(\Delta Z^1_{\tau_{i}})\Delta Z^1_{\tau_{i}} 
\end{equation}
\begin{remark}
Note that there is a distinction between the self-financing condition in Eq. (\ref{reb}) and the self-financing condition proposed by \c{C}etin \textit{et al.} (2004) and Guasoni \textit{et al.} (2010). The first authors  consider only the first part of the Eq. (\ref{reb}), i.e. $-M(\tau_{i},\Delta Z^1_{\tau_{i}})\Delta Z^1_{\tau_{i}}$. This can be easily seen by assuming that $A(\tau_{i},\Delta Z^1_{\tau_{i}})=B(\tau_{i},\Delta Z^1_{\tau_{i}})$. In a nutshell, they are just incorporating the liquidity risk into the self-financing equation, measured by the impact of the traded volume on the unique price. In doing so, the self-financing equation misses the implicit transaction costs generated over a given interval by trading at a different buy and sell price. The second authors instead  consider both terms of  Eq. (\ref{reb}), but neglecting the dependence of prices  on the traded volume. We can thus say that the self-financing condition proposed here accounts for both the implicit transaction costs generated by the liquidity risk and the bid-ask spread. 
\end{remark}
Using Eqs. (\ref{eq1}) and  (\ref{self}), the self-financing portfolio  becomes 
\begin{equation}
V_{\tau_{i}}^{Z^1} = Z^0_{\tau_{i-1}} -M(\tau_{i},\Delta Z^1_{\tau_{i}})\Delta Z^1_{\tau_{i}}-\frac{1}{2}P(\tau_{i},\Delta Z^1_{\tau_{i}})\mbox{sgn}(\Delta Z^1_{\tau_{i}})\Delta Z^1_{\tau_{i}}+ Z^1_{\tau_{i}} M(\tau_{i},0)
\end{equation}
which again is 
\begin{eqnarray}
V_{\tau_{i}}^{Z^1} &=& Z^0_{\tau_{i-1}}+Z^1_{\tau_{i-1}} M(\tau_{i-1},0) -M(\tau_{i},\Delta Z^1_{\tau_{i}})\Delta Z^1_{\tau_{i}}\nonumber\\&-&\frac{1}{2}P(\tau_{i},\Delta Z^1_{\tau_{i}})\mbox{sgn}(\Delta Z^1_{\tau_{i}})\Delta Z^1_{\tau_{i}}+ Z^1_{\tau_{i}} M(\tau_{i},0)-Z^1_{\tau_{i-1}} M(\tau_{i-1},0)
\end{eqnarray}
Recursively substituting, the value of the portfolio at time $T$ assumes the following form
\begin{eqnarray}
V_{T}^{Z^1} &=& V_{0}^{Z^1} -\sum_{i=1}^{n}M(\tau_{i},\Delta Z^1_{\tau_{i}})\Delta Z^1_{\tau_{i}}\nonumber\\&-& \sum_{i=1}^{n}\frac{1}{2}P(\tau_{i},\Delta Z^1_{\tau_{i}})\mbox{sgn}(\Delta Z^1_{\tau_{i}})\Delta Z^1_{\tau_{i}}+Z^1_{T} M(T,0)-Z^1_{0} M(0,0)
\end{eqnarray}
which after some manipulations becomes
\begin{eqnarray}\label{port}
V_{T}^{Z^1} &=& V_{0}^{Z^1}+\sum_{i=1}^{n}Z^1_{\tau_{i-1}}[M(\tau_{i},0)-M(\tau_{i-1},0)]\nonumber\\&-&\sum_{i=1}^{n}[M(\tau_{i},\Delta Z^1_{\tau_{i}})-M(\tau_{i},0)]\Delta Z^1_{\tau_{i}}\nonumber\\&+&\sum_{i=1}^{n}\frac{1}{2}Z^1_{\tau_{i-1}}[\mbox{sgn}(\Delta Z^1_{\tau_{i}})P(\tau_{i},0)-\mbox{sgn}(\Delta Z^1_{\tau_{i-1}})P(\tau_{i-1},0)]\nonumber \\&-&\sum_{i=1}^{n}\frac{1}{2}[P(\tau_{i},\Delta Z^1_{\tau_{i}})-P(\tau_{i},0)]\mbox{sgn}(\Delta Z^1_{\tau_{i}})\Delta Z^1_{\tau_{i}}\nonumber\\&+&\frac{1}{2}P(0,0)\mbox{sgn}(\Delta Z^1_{0})Z^1_{0}-\frac{1}{2}P(T,0)\mbox{sgn}(\Delta Z^1_{T})Z^1_{T}
\end{eqnarray}
where conventionally we set $Z^1_{u}$ equal to $Z^1_{0}$ for all $u<0$.
\begin{remark}\label{value}
The right-hand side of Eq. (\ref{port}) accounts for the portfolio's value in the implicit transaction cost economy. $Z^{0}_{0}+Z^1_{0}M(0,0)$ gives the initial value of the portfolio in the standard (implicit transaction cost) economy. The second term gives the capital gain in the standard economy (\ref{app1}). This term together with 
\begin{eqnarray}
&&\sum_{i=1}^{n}\frac{1}{2}Z^1_{\tau_{i-1}}[\mbox{sgn}(\Delta Z^1_{\tau_{i}})P(\tau_{i},0)-\mbox{sgn}(\Delta Z^1_{\tau_{i-1}})P(\tau_{i-1},0)]\nonumber\\&&+\frac{1}{2}P(0,0)\mbox{sgn}(\Delta Z^1_{0})Z^1_{0}-\frac{1}{2}P(T,0)\mbox{sgn}(\Delta Z^1_{T})Z^1_{T}
\end{eqnarray}
give the capital gain in an economy with only bid-ask spread (e.g. Guasoni \textit{et al.} (2010)). It can be easily seen by writing it as 
\begin{eqnarray}\label{vl}
&-&\sum_{i=1}^{n}M(\tau_{i},0)\Delta Z^1_{\tau_{i}}-\sum_{i=1}^{n}\frac{1}{2}P(\tau_{i},0)\mbox{sgn}(\Delta Z^1_{\tau_{i}})\Delta Z^1_{\tau_{i}} \nonumber\\ &+&M(T,0)Z^1_{T}-M(0,0)Z^1_{0}
\end{eqnarray}
This means that if we assume that prices are independent on the traded volume, $Z^{0}_{0}+Z^1_{0}M(0,0)$ together with Eq. (\ref{vl}) give the self-financing portfolio. One also see that when the trading strategy $Z^1$ has bounded variation, the above equation has a continuous version. Note also that $\sum_{i=1}^{n}\frac{1}{2}$ $P(\tau_{i},0)\mbox{sgn}(\Delta Z^1_{\tau_{i}})\Delta Z^1_{\tau_{i}}$  gives the implicit transaction costs registered over the interval $[0,T]$ in an economy without price impact. The term $\sum_{i=1}^{n}[M(\tau_{i},\Delta Z^1_{\tau_{i}})-M(\tau_{i},0)]\Delta Z^1_{\tau_{i}}$ combined with the term $\sum_{i=1}^{n}\frac{1}{2}$ $[P(\tau_{i},\Delta Z^1_{\tau_{i}})-P(\tau_{i},0)]\mbox{sgn}(\Delta Z^1_{\tau_{i}})\Delta Z^1_{\tau_{i}}$ capture the so-called price impact, the price impact that an investor can create by trading on an asset (see, for example, Huberman \& Stanzl (2005) for a more detailed discussion of this concept). For instance, a positive value of these terms means that the prices move up when the investor buys, and \textit{vice} \textit{versa}. 
\end{remark}
\subsection{Continuous case}\label{cont}
\ref{app} shows the derivation of the self-financing portfolio in the continuous case. More specifically, for any sequence of random partitions  $\sigma_{n}: 0=\tau_{0}^n\leq\tau_{1}^n\leq...\leq \tau_{i_n}^n =t$ tending to the identity, $t\in[0,T]$, the self-financing portfolio in continuous time reads 
\begin{eqnarray}\label{zari}
V_{t}^{Z^1}&=&V_{0}^{Z^1}+\int_{0}^{t}Z^1_{s}dM(s,0)+\frac{1}{2}\int_{0}^{t}Z^1_{s}d\tilde{P}^{Z^1}(s,0)-\int_{0}^{t}M'(s,0)d[Z^1,Z^1]^{c}_{s}\nonumber\\&-&\sum_{0<s\leq t}[M(s,Z^1_{s}-Z^1_{s-})-M(s,0)](Z^1_{s}-Z^1_{s-}) \nonumber\\&-&\frac{1}{2}\int_{0}^{t}P'(s,0)\mbox{sgn}(\Delta Z^1_{s})d[Z^1,Z^1]^{c}_{s}\nonumber\\&-&\frac{1}{2}\sum_{0 <s \leq t}[P(s,Z^1_{s}-Z^1_{s-})-P(s,0)]\mbox{sgn}(Z^1_{s}-Z^1_{s-})(Z^1_{s}-Z^1_{s-})\nonumber\\&-&\frac{1}{2}P(t,0)\mbox{sgn}(Z^1_{t}-Z^1_{t-})Z^1_{t}
\end{eqnarray}
when $Z^1$ is c\'adl\'ag, and 
\begin{eqnarray}\label{zari1}
V_{t}^{Z^1}&=&V_{0}^{Z^1}+\int_{0}^{t}Z^1_{s}dM(s,0)+\frac{1}{2}\int_{0}^{t}Z^1_{s}d\tilde{P}^{Z^1}(s,0) -\int_{0}^{t}M'(s,0)d [Z^1,Z^1]^{c}_{s}\nonumber\\&-&\sum_{0\leq s< t}[M(s,Z^1_{s+}-Z^1_{s})-M(s,0)](Z^1_{s+}-Z^1_{s})\nonumber\\&-&\frac{1}{2}\int_{0}^{t}P'(s,0)\mbox{sgn}(\Delta Z^1_{s})d[Z^1,Z^1]^{c}_{s}\nonumber\\&-&\frac{1}{2}\sum_{0\leq s< t}[P(s,Z^1_{s+}-Z^1_{s})-P(s,0)]\mbox{sgn}(Z^1_{s+}-Z^1_{s})(Z^1_{s+}-Z^1_{s})
\end{eqnarray}
when $Z^1$ is c\'agl\'ad. 

The continuous version of the self-financing portfolio was proved by making the following assumption. 
\begin{assumption}\label{assnew}
We assume that $M(t,y)+\frac{1}{2}\mbox{sgn}(y)P(t,y)$ is non-decreasing in $y$. That is, $x\leq y$ implies $M(t,x)+\frac{1}{2}\mbox{sgn}(x)P(t,x)\leq M(t,y)+\frac{1}{2}\mbox{sgn}(y)P(t,y)$ a.s. $\mathbb{P}$, a.e. $t$.
\end{assumption}
An immediate consequence of the above assumption is that $B(t,y)\leq M(t,0)$ for $y<0$ and $A(t,y)\geq M(t,0)$ for $y> 0$. This observation says that the ask and the bid price are always greater or equal and smaller or equal than the mid-price in an economy without price impact.

The random process $\tilde{P}^{Z^1}(s,0)$ is defined as $\mbox{sgn}(\Delta Z^1_{s})P(s,0)$. Note that $\tilde{P}^{Z^1}$ is a locally bounded semimartingale. This is because $\mbox{sgn}(\Delta Z^1_{\cdot})$ has bounded variation on compacts due to the finite jumps assumption, and thus is a semimartingale. The result follows by noticing that the product of semimartingales is again a semimartingale. 

As noted in the two equations above, $Z^0$ is uniquely determined by $Z^1$. In particular, when $Z^1$ is c\'adl\'ag and continuous, $V^{Z^1}$ is c\'adl\'ag. Thus, $Z^0$ being optional is justified in this case. 
Note also that when the filtration $\mathbb{F}$ satisfies the usual conditions, every semimartingale has a 
c\'adl\'ag version.  Henceforth, semimartingales can be taken in their c\'adl\'ag versions. 
Similarly, choosing $Z^0$ l\'adl\'ag adapted is justified when $Z^1$ is c\'agl\'ad.   
\begin{remark}
The continuous version of the self-financing portfolio illustrates important differences with respect to the usual self-financing portfolios in the finance literature.  With regard to the literature analysing arbitrage pricing theory under transaction costs, it supposes that prices depend not only on time and randomness, but also on the traded volume. On the other side, it is different from the one proposed by  
\c{C}etin \textit{et al.} (2004) since now we are assuming different prices for buying and selling. Here, the sign of the transaction is important since it determines the price at which the investor trade. This also explain why we have additional terms 
in the self-financing portfolio. 
\end{remark}
\section{First fundamental theorem under transaction costs}\label{na}
The goal of the present section is to give conditions for the absence of arbitrage opportunities in the model illustrated in the previous sections. As a first step, we will first show that there is NA in the implicit transaction cost economy if there exists an  equivalent local martingale measure. Then, by focusing only on c\'agl\'ad adapted trading strategies of bounded quadratic variation and with finitely many jumps, we will prove that there is NFLVR in the implicit transaction cost economy if and only if  there exists an equivalent local martingale measure. 

Before presenting the results, it is necessary to introduce and formalize two key concepts, portfolio admissibility and arbitrage opportunity. 
\begin{definition}
The process $Z$ is called $\alpha$-admissible if the trading strategy $Z^1$ defined in Definition \ref{tradin} is $\alpha$-admissible. Keeping this in mind, a self-financing strategy $Z^1$ is $\alpha$-admissible if  there exists a constant $\alpha> 0$ such that
\begin{equation}\label{eqport2}
V_{0}^{Z^1}=0, \quad V_{t}^{Z^1}  \geq - \alpha  \qquad \mbox{a.s. } \mathbb{P},  \mbox{a.e $t \in [0, T]$}
\end{equation} 
\end{definition}
Intuitively, the idea of an arbitrage opportunity is relatively easy to understand - a portfolio is an arbitrage opportunity when it allows investors to make money for no risk. Formalizing the concept, however, is more complicated than one might think.  With regard to  arbitrage opportunities, one has to make a distinction between two standard concepts, NA and NFLVR. The distinction between NA and NFLVR is important since the former is not sufficient to exclude approximate arbitrage opportunities.

The reader familiar with the standard arbitrage pricing theory will at this point recall that the standard definition of an arbitrage is as follows. 
\begin{definition}
An $\alpha$-admissible strategy $Z^1$ is called an arbitrage on $[0, T]$ if
\begin{equation}\label{arb}
V_{0}^{Z^1}=0, \qquad V_{T}^{Z^1} \geq 0 \qquad \mbox{a.s. $\mathbb{P}$} \qquad \mbox{and} \qquad \mathbb{P} (V_{T}^{Z^1} >0)>0 
\end{equation}
A market is arbitrage-free if there are no $\alpha$-admissible trading strategies $Z^1$ satisfying  (\ref{arb}). 
\end{definition}
The goal now is to prove that the implicit transaction cost economy excludes trading strategies which satisfy the above definition. In order to prove the result, we will need the following lemma.
\begin{lemma}\label{transac}
Given Assumption \ref{assnew} and $Z^1$ c\'adl\'ag and c\'agl\'ad respectively, the processes
\begin{eqnarray}\label{t}
T_{t}^{Z^1}&=&\int_{0}^{t}M'(s,0)d[Z^1,Z^1]^{c}_{s}+\sum_{0<s\leq t}[M(s,Z^1_{s}-Z^1_{s-})-M(s,0)](Z^1_{s}-Z^1_{s-}) \nonumber\\&+&\frac{1}{2}\int_{0}^{t}P'(s,0)\mbox{sgn}(\Delta Z^1_{s})d[Z^1,Z^1]^{c}_{s}\nonumber\\&+&\frac{1}{2}\sum_{0 <s \leq t}[P(s,Z^1_{s}-Z^1_{s-})-P(s,0)]\mbox{sgn}(Z^1_{s}-Z^1_{s-})(Z^1_{s}-Z^1_{s-})
\end{eqnarray}
and
\begin{eqnarray}\label{t1}
T_{t}^{Z^1}&=&\int_{0}^{t}M'(s,0)d[Z^1,Z^1]^{c}_{s}+\sum_{0\leq s< t}[M(s,Z^1_{s+}-Z^1_{s})-M(s,0)](Z^1_{s+}-Z^1_{s}) \nonumber\\&+&\frac{1}{2}\int_{0}^{t}P'(s,0)\mbox{sgn}(\Delta Z^1_{s})d[Z^1,Z^1]^{c}_{s}\nonumber\\&+&\frac{1}{2}\sum_{0\leq s< t}[P(s,Z^1_{s+}-Z^1_{s})-P(s,0)]\mbox{sgn}(Z^1_{s+}-Z^1_{s})(Z^1_{s+}-Z^1_{s})
\end{eqnarray}
are non-negative and non-decreasing on the interval $[0,T]$. 
\end{lemma}
\begin{proof}
The fact that  $T^{Z^1}_{t}$ is non-negative and non-decreasing in $t$ follows  directly from Assumption \ref{assnew}.  
\end{proof}
Note that the process $T^{Z^1}_{t}$ gives the price impact in the continuous time setting.  

We are now ready to state the following theorem which is crucial in the arbitrage pricing theory.
\begin{lemma}\label{th1}
If there exists a probability $\mathbb{Q}$, equivalent to $\mathbb{P}$, such that $M(\cdot,0)$ is a $\mathbb{Q}$-local martingale, then the implicit transaction cost economy has NA.   
\end{lemma}
\begin{proof}
The proof of the theorem is derived by using classical arguments as in the standard theory of arbitrage pricing. Let $\mathbb{Q}$ be an equivalent local martingale measure and suppose there exist an arbitrage trading strategy $Z^1$ such that $V_{0}^{Z^1}=0,  V_{T}^{Z^1} \geq 0$ a.s. $\mathbb{P}$ and  $\mathbb{P}(V_{T}^{Z^1} >0)>0$. Then, since $\mathbb{Q}$ is equivalent with respect to $\mathbb{P}$, one easily deduces that also $V_{T}^{Z^1}\geq 0$  a.s. $\mathbb{Q}$  and $\mathbb{Q}(V_{T}^{Z} >0)>0$. Recalling that $Z^1$ is $\alpha$-admissible we have that $V_{t}^{Z^1}=\int_{0}^{t}Z^1_{s}dM(s,0)+\frac{1}{2}\int_{0}^{t}Z^1_{s}d\tilde{P}^{Z^1}(s,0)-T_{t}^{Z^1}-\frac{1}{2}P(t,0)\mbox{sgn}(Z^1_{t}-Z^1_{t-})Z^1_{t}\geq -\alpha$ for some $\alpha>0$. Using the fact that $T_{t}^{Z^1}$ and $-\frac{1}{2}\int_{0}^{t}Z^1_{s}d\tilde{P}^{Z^1}(s,0)+\frac{1}{2}P(t,0)\mbox{sgn}(Z^1_{t}-Z^1_{t-})Z^1_{t}$ are non-negative  for every self-financing trading strategy and $t\in[0,T]$ (see Lemma \ref{transac} for what concerns $T_{t}^{Z^1}$), we have also that $\int_{0}^{t}Z^1_{s}dM(s,0)$ is bounded from below by the constant $-\alpha$. It follows that  $\int_{0}^{t}Z^1_{s}dM(s,0)$ is a $\mathbb{Q}$-local martingale. One can then show using Fatou lemma that, since $\int_{0}^{t}Z^1_{s}dM(s,0)$ is a local martingale,  it is also  a $\mathbb{Q}$-supermartingale. Moreover, the processes $T_{t}^{Z^1}$ and $-\frac{1}{2}\int_{0}^{t}Z^1_{s}d\tilde{P}^{Z^1}(s,0)+\frac{1}{2}P(t,0)\mbox{sgn}(Z^1_{t}-Z^1_{t-})Z^1_{t}$ besides being non-negative are also non-decreasing for every self-financing trading strategy. This implies that $V_{t}^{Z^1}$ is a supermartingale. Therefore,   $0\geq E_{\mathbb{Q}}(V_{T}^{Z^1})$. This is a contradiction to the arbitrage definition and thus showing that an arbitrage  cannot exists. The c\'agl\'ad case works in the same manner. 
\end{proof}
Putting the focus on c\'agl\'ad adapted processes $Z^1$ with bounded quadratic  variation and of finitely many jumps, the rest of this section provides necessary and sufficient conditions for the implicit transaction cost economy  to be free of arbitrage opportunities, thus excluding free lunch with vanishing risk (FLVR). Moreover, it provides an extension of the FFTAP to the economy with implicit transaction costs. Before proceeding further, we will first provide a definition of what we mean by FLVR. 
\begin{definition}\label{def4}
A FLVR  is  a sequence of $\alpha_{n}$-admissible trading strategies $Z^1_{n}$, where $V_{T}^{Z^1_{n}}\geq-\alpha_{n}$, $V^{Z^1_{n}}_{0}=0$, and the sequence $\alpha_{n}$ tends to zero, we have that $V_{T}^{Z^1_{n}}$ converges a.s. to a non-negative random variable $V$, with $V$ not identically zero. We say the implicit transaction cost economy or $M(\cdot,0)$ satisfies the NFLVR property if for any sequence of $\alpha_{n}$-admissible trading strategies $Z^1_{n}$, $V^{Z^1_{n}}_{0}=0$, and the sequence $\alpha_{n}$ tends to zero, we have $V_{T}^{Z^1_{n}}$ converges a.s. to some non-negative limit $V$, then $V = 0$ a.s..
\end{definition}
Let us recall an important result that will be used often in the rest of the paper. 
\begin{proposition}\label{proe}
Let $\theta_{n}$ be a sequence of $[0,\infty)$-valued measurable random variables on a probability space $(\Omega,\mathcal{F},\mathbb{P})$. There exists $\psi_{n}\in convex(\theta_{n},\theta_{n+1},...)$ a sequence, such that $\psi_n$ converges a.s. to a $[0,\infty]$-valued function $\psi$. Moreover, if $convex(\theta_{n},\theta_{n+1},...)$ is bounded in $L^0$, then $\psi<\infty$ a.s.. If $\mathbb{P}(\theta_n>\beta)>\delta$ for all $n$ with $\beta>0$ and $\delta>0$, then $\mathbb{P}(\psi>0)>0$.   
\end{proposition} 
\begin{proof}
See Delbaen \& Schachermayer (1994), Lemma A1.1.
\end{proof}
Now we proceed with the proof of the FFTAP. 
\begin{lemma}\label{lemmbou}
Define the set $\mathcal{K}_{0}=\{V_{T}^{Z^1}| Z^1 \quad\text{is}\quad \alpha-\text{admissible}\}$. If $M(\cdot,0)$ satisfies NFLVR then the set $\mathcal{K}_{0}$ is bounded in $L^{0}$, that is, in probability.
\end{lemma}
\begin{proof}
By definition, $\mathcal{K}_{0}$ contains the trading strategies which are continuous, $\alpha$-admissible and of bounded variation, it is therefore convenient first to prove the result for these particular strategies. The arguments of the proof are similar to those used by Delbaen \& Schachermayer (1994) in Proposition 3.1. By contradiction, suppose that the set $\{V_{T}^{Z^1}| Z^1 \quad \text{is $\alpha$-admissible continuous of bounded variation}\}$  is not bounded in $L^0$. This means there exists a sequence $Z^1_{n}$ of $\alpha$-admissible, continuous of bounded variation  integrands and $\epsilon > 0$ such that $P(V_{T}^{Z^1_{n}}\geq \frac{\alpha}{\alpha_{n}})>\epsilon$ with $\alpha_{n}>0$ and $\lim_{n\rightarrow \infty}\alpha_{n}=0$. Take now the sequence $H_n$ defined as $V_{T}^{Z^1_{n}}\frac{\alpha_{n}}{\alpha}\wedge 1$ which satisfies $H_n \geq -\alpha_{n}$, and $P(H_n = 1) > \epsilon$. Note that the $H_{n}$ are self-financing portfolios. Then, Proposition 1 provides us with a sequence of convex combinations of $H_n$ converging a.s. to $H\in [0, 1]$ with $E(H) \geq \epsilon$, $P(H > 0) = \beta\geq \epsilon > 0$, therefore a contradiction. 

To prove the final result it is then sufficient to show that the boundedness in $L^0$ of $\mathcal{K}_{0}$, restricted to continuous of bounded variation trading strategies, implies that the set $\mathcal{K}_{0}$ is also bounded for the other trading strategies. Now, suppose by contradiction that $\mathcal{K}_{0}$ restricted to c\'adl\'ag and c\'agl\'ad $\alpha$-admissible trading strategies is not bounded in $L^0$. Then, there is an $\alpha$-admissible $Z^1$, $\epsilon >0$, such that $P(V_{T}^{Z^1}\geq c)> \epsilon$ with $c>0$. By the non-positivity of $-T_{T}^{Z^1}$, it easily follows that also $P(\int_{0}^{T}Z^1_{s}dM(s,0)+\frac{1}{2}\int_{0}^{T}Z^1_{s}d\tilde{P}^{Z^1}(s,0)\geq c)> \epsilon$. We are supposing for simplicity that $Z^1$ is c\'agl\'ad. Nothing changes if instead we assume $Z^1$  is c\'adl\'ag. Now, using the fact that stochastic integrals of the form $\int_{0}^{T}Z^1_{s}dM(s,0)$ and $\frac{1}{2}\int_{0}^{T}Z^1_{s}d\tilde{P}^{Z^1}(s,0)$ can be approximated (uniformly on compacts in probability) by stochastic integrals with continuous and bounded variation integrands (see Bank \& Baum (2004) and \c{C}etin et.al (2004)) of the form $\int_{0}^{T}Z^1_{n, s}dM(s,0)$ and $\frac{1}{2}\int_{0}^{T}Z^1_{n, s}d\tilde{P}^{Z^1_{n}}(s,0)$, with $Z^1_{n, s}=n\int_{s-\frac{1}{n}}^{s}Z^1_{u}du$, $Z^1_{n}\rightarrow Z^1$ a.s. for every $s\geq 0$ and $Z^1_{u}=Z^1_{0}$ for $u<0$, as $n\rightarrow \infty$, we can find a subsequence $Z^1_{n_{k}}$ of bounded variation continuous integrands, $n_{k}\rightarrow \infty$ as $k\rightarrow \infty$, such that $\lim_{k\rightarrow \infty}\int_{0}^{T}Z^1_{n_{k}, s}dM(s,0)=\int_{0}^{T}Z^1_{s}dM(s,0)$ and $\lim_{k\rightarrow \infty}\frac{1}{2}\int_{0}^{T}Z^1_{n_{k}, s}d\tilde{P}^{Z^1_{n_{k}}}(s,0)=\frac{1}{2}\int_{0}^{T}Z^1_{s}d\tilde{P}^{Z^1}(s,0)$ a.s.. Note that $\tilde{P}^{Z^1_{n}}$ converges  under the semimartingale topology to $\tilde{P}^{Z^1}$ and since $Z^1_{n}$ converges a.s. to $Z^1$, we have that $\frac{1}{2}\int_{0}^{T}Z^1_{n, s}d\tilde{P}^{Z^1_{n}}(s,0)$ converges u.c.p. to $\frac{1}{2}\int_{0}^{T}Z^1_{ s}d\tilde{P}^{Z^1}(s,0)$. Obviously, the $T_{T}^{Z^1_{n_{k}}}$ are equal to zero due to the continuity and bounded variation of the trading strategies and the $Z^1_{n_{k}}$ are $\alpha$-admissible for some $\alpha$ and for every $k$. Therefore, $P[\lim_{k\rightarrow \infty}(\int_{0}^{T}Z^1_{n_{k}, s}dM(s,0)+\frac{1}{2}\int_{0}^{T}Z^1_{n_{k}, s}d\tilde{P}^{Z^1_{n_{k}}}(s,0))\geq c]> \epsilon$. Using the fact that an a.s. convergent sequence is always a.s. bounded, we obtain a contradiction to the boundedness of the admissible self-financing portfolios with continuous and bounded variation trading strategies. This ends the proof. 
\end{proof}
Our next goal is to prove the closure, in the Fatou sense, of the following set $\mathcal{C}_{0}=\mathcal{K}_{0}-L^{0}_{+}$. To prove this crucial result we shall need before to state and prove two lemmas. It is worth pointing out here that the lemma above and the following two lemmas hold for general $\alpha$-admissible strategies satisfying Definition \ref{tradin}. 
\begin{lemma}\label{lemmot}
Let $M(\cdot,0)$ satisfy NFLVR. Then  $\{[Z^1,Z^1]^{c}_{T}| Z^1 \quad \text{is}\quad \alpha-\text{admissible}\}$  is bounded in $L^{0}$.
\end{lemma} 
\begin{proof}
This is a consequence of Lemma \ref{lemmbou}. Let $Z^1$ be an $\alpha$-admissible  trading strategy. 
By Assumption \ref{mid}, $M'(\cdot,0)+\frac{1}{2}P'(\cdot,0)\mbox{sgn}(\Delta Z^1_{\cdot})$ is finite on $[0,T]$ and, therefore the infimum is attained. Furthermore, Assumption \ref{assnew} implies that $M'(\cdot,0)+\frac{1}{2}P'(\cdot,0)\mbox{sgn}(\Delta Z^1_{\cdot})$ is non-negative. Therefore,  
\begin{eqnarray}
&&0\leq [Z^1,Z^1]^{c}_{T}\leq \int_{0}^{T}\inf_{s\in[0,T]}[M'(s,0)+\frac{1}{2}P'(s,0)\mbox{sgn}(\Delta Z^1_{s})]d[Z^1,Z^1]^{c}_{s}\nonumber\\&&\leq\int_{0}^{T}[M'(s,0)+\frac{1}{2}P'(s,0)\mbox{sgn}(\Delta Z^1_{s})]d[Z^1,Z^1]^{c}_{s}\nonumber\\&&\leq  T_{T}^{Z^1}\leq \alpha+V_{T}^{Z^1}+T_{T}^{Z^1}\leq(\alpha + g) + V_{T}^{Z^1}
\end{eqnarray}
where $g>0$ gives the boundedness constant of $T_{T}^{Z^1}$ which by \ref{app} is finite. Lemma \ref{lemmbou} then implies that  $[Z^1,Z^1]^{c}_{T}$ is bounded in $L^0$ for every $\alpha$-admissible $Z^1$ trading strategy. 
\end{proof}
As another interesting result, we can deduce from Lemma \ref{lemmbou} that the set $\{||Z^1||_{T}|  \\\text{$Z^1$} {\mbox{\quad is $\alpha$-admissible continuous and of bounded variation}}\}$ is bounded in $L^0$, where $||\cdot||_{\cdot}$ in this case denotes the pathwise total variation (as the supremum taken over all the partitions). To see this, write $V_{T}^{Z^1}=\int_{0}^{T}Z^1_{s}dM(s,0) +\frac{1}{2}\int_{0}^{T}Z^1_{s}d\tilde{P}^{Z^1}(s,0)$ as
\begin{equation}
V_{T}^{Z^1}=-\int_{0}^{T}M(s,0)dZ^1_{s}-\frac{1}{2}\int_{0}^{T}P(s,0)d||Z^1||_{s}+M(T,0)Z^1_{T}-M(0,0)Z^1_{0}
\end{equation}
and we refer to Guasoni \textit{et. al} (2010) for more details on this latter equation. Note  that  in Guasoni \textit{et. al} (2010) it is assumed that $M(T,0)Z^1_{T}-M(0,0)Z^1_{0}$ is equal to zero. Now, since $V_{T}^{Z^1}$ is bounded in $L^0$ we can prove similarly as in the above lemma that also the set containing the elements $||Z^1||_{T}$ with $Z^1$, $\alpha$-admissible, continuous and of bounded variation 
is bounded in $L^0$. 
\begin{remark}
It is noteworthy mentioning here that implicit transaction costs cannot be avoided by following a continuous trading strategy of bounded variation. That is, the transaction costs cannot be  avoided with smarter (smooth) trading strategies.  The best that an investor can do when trading in a market with transaction costs is to use continuous trading strategies of bounded variation to avoid the price impact (Bank \& Baum (2004) and \c{C}etin et.al (2004)). This is because these trading strategies induce no path dependency in the evolution of the stock prices. Hence, splitting a large trade into infinitesimal smaller ones reduces the price impact. 
\end{remark} 
The following compactness lemma for $\alpha$-admissible trading strategies is similar to that obtained by Campi \& Schachermayer (2006) and Guasoni (2002) for predictable of bounded variation processes. 
\begin{lemma}\label{lemmanew}
Let $Z^{1}_{n}$ be a sequence of $\alpha$-admissible trading strategies. Suppose further that $M(\cdot,0)$ satisfies NFLVR. Then, there exists a sequence $\mathbb{Z}^{n}\in convex(Z^1_{n},Z^1_{n+1},...)$ such that $\mathbb{Z}^n$ converges a.s. $\omega$ for each $t\in[0,T]$ to a  process $\hat{W}$ which is adapted, c\'agl\'ad with finitely many jumps, and has bounded variation and zero quadratic variation.
\end{lemma}
\begin{proof}
By assumption $Z^{1}_{n}$ are predictable of bounded quadratic variation processes. We can then use the result in \c{C}etin \textit{et. al} (2004) again, to show that $Z^{1}_{n}$ can be obtained as the a.s. pointwise limit, as $m\rightarrow\infty$, of the following process
\begin{equation}
G^{m}_{n,s} =m\int_{s-\frac{1}{m}}^{s}Z^1_{n,u}du
\end{equation}
for all $s\geq0$ and $Z^1_{n, u}=Z^1_{n, 0}$ for $u<0$ and $n\geq 1$. We use this observation to show that $\mbox{convex}(Z^1_{n,t},Z^1_{n+1,t},...)$ is bounded in $L^0$. 

Now note that $G^{m}_{n}$ are $\alpha$-admissible continuous and of bounded variation for every $m,n\geq 1$, and thus using the NFLVR property the $||G^{m}_{n}||_{T}$ are bounded in $L^0$ by previous results. Hence, this implies that
\begin{equation}\label{lali}
|G^{m}_{n,t}-G^{m}_{n,0}|=|\sum_{i=1}^{n}(G^{m}_{n,t_i}-G^{m}_{n,t_{i-1}})|\leq ||G^{m}_{n}||_{t}\leq||G^{m}_{n}||_{T}
\end{equation} 
for every $t\in[0,T]$ and $0\leq t_0\leq t_1\leq...\leq t_n \leq t$. 

Given a convex combination $(\beta_j)_{j=n}^{\infty}$, we write $\sum_{j\geq n} \beta_{j}Z^1_{j, t} $ as
\begin{eqnarray}\label{erin}
\sum_{j\geq n} \beta_{j}Z^1_{j, t} =\lim_{m\rightarrow\infty}\sum_{j\geq n}\beta_{j}G^{m}_{j,t}
\end{eqnarray}
a.s.. The interchange of the limit is possible since $\lim_{m \rightarrow \infty} G^{m} _{j,t} = Z^1_
{j,t}$ and $G^{m} _{j,t}$ besides being bounded in $L^0$ for every $m, j$ and $t \in [0, T]$ are continuous and of bounded variation, and thus a.s. bounded over $[0, T]$. In addition, using the results of Lemma 4.3 in Guasoni (2002), it can also be easily shown that the set composed with the elements $||G^{m}_{n}||_{T}$ is a.s. bounded for each $m,n\geq 1$. The final result follows then from (\ref{lali}). Then, since $G^{m}_{j,t}$ are uniformly bounded for every $m$, $j$ and $t\in[0,T]$ so is the limit. It follows that $\mbox{convex}(Z^1_{n,t}, Z^1_{n+1,t},...)$ is bounded in $L^0$ (also a.s. bounded) for each $t\in[0,T]$. Therefore by Proposition \ref{proe} together with a diagonalization argument  there exist convex weights $(\beta_j)_{j=n}^{\infty}$ such that 
\begin{eqnarray}\label{mu}
\mathbb{W}^{n,m}_{t}&=&\sum_{j\geq n} \beta_{j}G^{m}_{j,t}\\\mathbb{C}^{n,m}_{t}&=&\sum_{j\geq n} \beta_{j}||G^{m}_{j}||_{t}
\end{eqnarray}
converge a.s. as $n\rightarrow\infty$, for every $t\in D:=([0,T]\cap\mathbb{Q})\cup\{T\}$ and $m\geq1$, to a random variable $\hat{W}^{m}_{t}$ and $\hat{C}^{m}_{t}$.  

Let $m_0\in\mathbb{N}$. Denote by $\tilde{\Omega}$ the event where both $\mathbb{W}^{n,m_0}_{t}\rightarrow\hat{W}^{m_0}_{t}$ and $\mathbb{C}^{n,m_0}_{t}\rightarrow \hat{C}^{m_0}_{t}$ a.s., $t\in D$, are true, so that $\mathbb{P}(\tilde{\Omega})=1$. Clearly, $q \rightarrow \hat{C}^{m_0}_{q}(w)$ is increasing over $D$, so that we can set $C^{m_0}_{0,t}=\inf\{\hat{C}^{m_0}_{q}|q\in\mathbb{Q}, q>t\}$ for every $t\in[0, T)$ and $C^{m_0}_{0,T}=\hat{C}^{m_0}_{T}$ on $\tilde{\Omega}$. Obviously, $C^{m_0}_{0}$ is right-continuous and non-decreasing.

Let  $\omega\in\tilde{\Omega}$ and $t\in ]0,T[$ be a point of continuity of $C^{m_0}_{0}$. By definition of $C^{m_0}_{0}$, for any $\epsilon>0$ there exist $p_1, p_2 \in\mathbb{Q}$ such that $p_1<t<p_2$ and $C^{m_0}_{0,p_1}(w)> C^{m_0}_{0,t}(w)-\epsilon$, $C^{m_0}_{0,p_2}(w)< C^{m_0}_{0,t}(w)+\epsilon$.  Then, using again the definition of $C^{m_0}_{0}$, there exist rationals $r_1$, $r_2$ and $N\in\mathbb{N}$ with  $p_1<r_1<t<r_2<p_2$ such that for all $n\geq N$, $\mathbb{C}^{n,m_0}_{r_1}(w)> C^{m_0}_{0,t}(w)-2\epsilon$ and $\mathbb{C}^{n,m_0}_{r_2}(w)< C^{m_0}_{0,t}(w)+2\epsilon$. This implies that $\mathbb{C}^{n,m_0}_{r_2}(w)-\mathbb{C}^{n,m_0}_{r_1}(w)= (\mathbb{C}^{n,m_0}_{r_2}(w)-C^{m_0}_{0,t}(w))+(C^{m_0}_{0,t}(w)-\mathbb{C}^{n,m_0}_{r_1}(w)) < 4\epsilon$ and 
\begin{equation}
|\mathbb{W}^{n,m_0}_{t}(w)-\mathbb{W}^{n,m_0}_{r_1}(w)|\leq \mathbb{C}^{n,m_0}_{t}(w)-\mathbb{C}^{n,m_0}_{r_1}(w)\leq\mathbb{C}^{n,m_0}_{r_2}(w)-\mathbb{C}^{n,m_0}_{r_1}(w)< 4\epsilon
\end{equation}
where the second inequality follows by the increasing property of the total variation. Since $\epsilon$ was arbitrary and $\mathbb{W}^{n,m_0}(w)$ converges on the rationals $\mathbb{Q}\cap[0,T]$, it follows that $\mathbb{W}^{n,m_0}_{t}(w)$ converges a.s. to a limit denoted by $\hat{W}^{m_0}_{t}(w)$ for each $t$ point of continuity of $C^{m_0}_{0}$ and $w\in\tilde{\Omega}$.  Therefore, the point of discontinuities of $C^{m_0}_{0}$ are a countable set. By taking convex combinations one can assume w.l.o.g. that $\mathbb{W}^{n,m_0}_{\tau_k}$ converges a.s. to $\hat{W}^{m_0}_{\tau_k}$ for each point of discontinuity $\tau_k$ of $C^{m_0}_{0}$. Thus, $\mathbb{W}^{n,m_0}_{t}$ converges a.s. to $\hat{W}^{m_0}_{t}$ for each $t\in[0,T]$. 

Obviously, $\mathbb{W}^{n,m_0}$ has bounded variation and zero quadratic variation for each $n\geq1$ and $t\in[0,T]$. The first property  is easy to prove. The second follows by the definition of the  quadratic variation, that is
\begin{eqnarray}
&&\lim_{k\rightarrow\infty}\sum_{i\geq1}[\sum_{j\geq n}\beta_{j}(G_{j, \tau_{i}^{k}}^{m_0}-G_{j, \tau_{i-1}^{k}}^{m_0})]^{2}=\sum_{j\geq n}\beta_{j}^{2}[G_{j}^{m_0},G_{j}^{m_0}]_{t}\nonumber\\&&+\lim_{k\rightarrow\infty}\sum_{i\geq1}\sum_{\substack{j, h\geq n\\j\neq h}}\beta_{j}\beta_{h}(G_{j, \tau_{i}^{k}}^{m_0}-G_{j, \tau_{i-1}^{k}}^{m_0})(G_{h, \tau_{i}^{k}}^{m_0}-G_{h, \tau_{i-1}^{k}}^{m_0})
\end{eqnarray}
where $\sigma_{k}: 0=\tau_{0}^k\leq\tau_{1}^k\leq...\leq \tau_{i_k}^k =t$ is a sequence of random partitions tending to the identity. Then, it is easy to see that also the limit $\hat{W}^{m_0}$ has bounded variation and zero quadratic variation. Since $m$  was arbitrary, we conclude that
\begin{equation}
\lim_{n\rightarrow\infty}\lim_{m\rightarrow\infty}\sum_{j\geq n}\beta_{j}G^{m}_{j,t}=\lim_{n\rightarrow\infty}\sum_{j\geq n}\beta_{j}Z^1_{j, t}
\end{equation}
exists a.s., and is equal to a limit $\hat{W}_{t}$, for each $t\in[0,T]$, process of bounded variation, and of zero quadratic variation. 

Clearly, the a.s. limit $\hat{W}$ is adapted and c\'agl\'ad. This is easy seen since the limit process $\hat{W}_{t}$ is obtained as $\hat{W}_{t}=\lim_{\substack{s<t, s\in\mathbb{Q}\\s\rightarrow t}} \lim_{n\rightarrow\infty} \lim_{m\rightarrow\infty}\mathbb{W}^{n,m}_{s}$. Then, one easily finds that $\hat{W}$ is adapted, has right limit, and is continuous from the left. Moreover, $\hat{W}$ is a process with finitely many jumps. This completes the proof.
\end{proof}
\begin{remark}\label{rem4}
Lemma \ref{lemmanew} allows us to assume, up to a sequence of convex combinations, that the sequence $Z^1_{n}$ converges a.s. to the random variable $\hat{W}$. However, it provides a more important result than the one stated. In fact, one can prove in the same vein that there exists a subsequence of $\alpha$-admissible strategies $Z^1_{n_{k}}$, $n_{k}\rightarrow \infty$ as $k\rightarrow \infty$, that converges a.s., as $k\rightarrow \infty$, to a c\'agl\'ad adapted process, still denoted by $\hat{W}$, which has finitely many jumps, bounded variation, and zero quadratic variation. The result can be proved by using the fact that $G^{m}_{n,t}$ are a.s. bounded for every $m,n\geq 1$ and $t\in[0,T]$, and that the set with elements $||G^{m}_{n}||_{T}$ is a.s. bounded for every $m, n\geq 1$. An application of Helly's theorem gives then the desired result. Another consequence of Lemma \ref{lemmanew} is that $[Z^1_{n_{k}},Z^1_{n_{k}}]^{c}_{s}$ converges pointwise to $0$ as $k\rightarrow\infty$, for every $s\in[0,T]$, when $Z^1_{n_{k}}$ converges a.s. to $\hat{W}$. This follows directly from the properties of $[Z^1_{n_{k}},Z^1_{n_{k}}]^{c}$, and  the fact that the $[Z^1_{n_{k}},Z^1_{n_{k}}]^{c}_{T}$ are bounded in $L^0$ and a.s. bounded when $M(\cdot,0)$ satisfies NFLVR. The proof that the set with elements $[Z^1_{n_{k}},Z^1_{n_{k}}]^{c}_{T}$ is a.s. bounded for every $n$ follows the same lines as that of Lemma 4.3 in Guasoni (2002).    
\end{remark}
We have as a corollary to Lemma \ref{lemmanew} the following result.  
\begin{corollary}\label{mjaft}
Let $Z^{1}_{n}$ be a sequence of continuous, $\alpha$-admissible and bounded variation trading strategies. Suppose in addition that $M(\cdot,0)$ satisfies NFLVR. Then, there exists a sequence $\mathbb{Z}^{n}\in convex(Z^1_{n},Z^1_{n+1},...)$ such that $\mathbb{Z}^n$ converges a.s. $\omega$ for each $t\in[0,T]$ to a  process $\hat{W}$ which is c\'agl\'ad, adapted with finitely many jumps, and has bounded variation and zero quadratic variation. 
\end{corollary}
\begin{theorem}\label{themu}
Let $\mathcal{K}^{r}_{0}=\{V_{T}^{Z^1}| Z^1 \quad \mbox{is} \quad \mbox{c\'agl\'ad}\quad  and \quad \alpha-\text{admissible}\}$. If $M(\cdot,0)$ satisfies NFLVR, then the set $\mathcal{C}^{r}_{0}=\mathcal{K}^{r}_{0}-L^{0}_{+}$ is Fatou-closed. 
\end{theorem}
\begin{proof}
To prove that $\mathcal{C}^{r}_{0}$ is Fatou-closed, it is sufficient to consider $(g_n)_{n\geq 1}$ in $\mathcal{C}^{r}_{0}$ such that $g_n\geq -\alpha$, $\alpha> 0$, and $\lim_{n\rightarrow \infty}g_n=g$ a.s.. The theorem  is then proved by showing that $g\in\mathcal{C}^{r}_{0}$. This in turn means finding an element $f\in\mathcal{K}^{r}_{0}$ such that $g\leq f$ a.s..  

By assumption, $g_n=\int_{0}^{T}Z^1_{n,s}dM(s,0)+\frac{1}{2}\int_{0}^{T}Z^1_{n,s}d\tilde{P}^{Z^1_{n}}(s,0)-T_{T}^{Z^1_{n}}\geq-\alpha$ for $\alpha$-admissible processes $Z^1_{n}$. Since the implicit transaction cost economy satisfies NFLVR and the $Z^1_{n}$ are $\alpha$-admissible, the set $\mathcal{K}^{r}_{0}$ and $([Z^1_{n},Z^1_{n}]^{c}_{T})_{n\geq1}$ for every $n\geq1$ are bounded in $L^0$ by Lemma \ref{lemmbou} and Lemma \ref{lemmot}. 

Next, note that also $\int_{0}^{t}Z^1_{n,s}dM(s,0)\geq -\alpha$ for every $n$ and $t\in[0,T]$. Then, the proof of Theorem 4.2 in Delbaen \& Schachermayer (1994) shows the existence of $\alpha$-admissible integrands, $L^{n}\in \mbox{convex}(Z^1_{n},Z^1_{n+1},...)$, in the standard economy, such that $\int_{0}^{t}L^{n}_{s}dM(s,0)$ converges in the semimartingale topology for all $t\in[0,T]$. Furthermore, the limit must necessarily be of the form $\int_{0}^{t}L_sdM(s,0)$. It follows that $L$ is $\alpha$-admissible in the standard economy and that the a.s. limit of $\int_{0}^{T}L_{s}^{n}dM(s,0)$ equals the a.s. limit of $\int_{0}^{T}Z^1_{n,s}dM(s,0)$.

Moreover, Lemma \ref{lemmanew} shows that there exists a convex combination $\mbox{convex}$ $(Z^1_{n},Z^1_{n+1},...)$ converging to $\hat{W}$. On the other hand, this lemma combined with Remark \ref{rem4} show that there exists a subsequence $Z^1_{n_{k}}$, $n_{k}\rightarrow \infty$ as $k\rightarrow \infty$, such that  $Z^1_{n_{k}}$ converges a.s. to an adapted c\'agl\'ad process with bounded variation, zero quadratic variation and finitely many jumps. We denote this process still by $\hat{W}$. Again, Lemma \ref{lemmanew} and Remark \ref{rem4} show that $\sup_{s\in[0,T]}(|Z^1_{n_{k},s}|)$ is finite a.s. for every $k\geq 1$. Using the dominated convergence theorem for stochastic integrals, we have that $\int_{0}^{T}Z^1_{n_{k}}dM(s,0)\rightarrow \int_{0}^{T}\hat{W}_{s}dM(s,0)$ in u.c.p. as $k\rightarrow \infty$. Then, by the uniqueness of the limit, $\int_{0}^{T}Z^{1}_{n,s}dM(s,0)$ converges a.s. to $\int_{0}^{T}\hat{W}_{s}dM(s,0)$.

Using again Lemma \ref{lemmanew}, there exists thus a subsequence of $Z^1_{n}$, denoted by $Z^1_{n_{k}}$,  which converges a.s., as $k\rightarrow \infty$, to the c\'agl\'ad adapted process $\hat{W}$ of zero quadratic variation and finitely many jumps. Then, by Remark \ref{rem4}, $[Z^1_{n_{k}},Z^1_{n_{k}}]^{c}$ converges to $0$ a.s. for every $t\in[0,T]$.  Lemma \ref{lemmanew} and Remark \ref{rem4} together with Lemma \ref{iku}, Lemma \ref{iku2} and Remark \ref{iku3} in  \ref{app2} show then that $T_{T}^{Z^1_{n_{k}}}$ converges a.s. to $T_{T}^{\hat{W}}$. Note that this holds for all $t\in[0,T]$. 

Take now $\frac{1}{2}\int_{0}^{T}Z^1_{n,s}d\tilde{P}^{Z^1_{n}}(s,0)$. Since these are non-positive random variables, by Proposition \ref{proe} we can assume that $\frac{1}{2}\int_{0}^{T}Z^1_{n,s}d\tilde{P}^{Z^1_{n}}(s,0)$ converges up to a sequence of convex combinations to a random variable $\eta$ with  $\eta<\infty$ a.s..  Then, given that $Z^1_{n_{k}}$ converges a.s. to $\hat{W}$ and that $\tilde{P}^{Z^1_{n_{k}}}$ converges under the semimartingale topology to $\tilde{P}^{\hat{W}}$, we get that also $\frac{1}{2}\int_{0}^{T}Z^1_{n,s}d\tilde{P}^{Z^1_{n}}(s,0)$ converges a.s. to $\frac{1}{2}\int_{0}^{T}\hat{W}_{s}d\tilde{P}^{\hat{W}}(s,0)$. 

Therefore, we have found a trading strategy $\hat{W}$ and a subsequence $Z^1_{n_{k}}$ such that
\begin{eqnarray}\label{fund}
f&=&\int_{0}^{T}\hat{W}_{s}dM(s,0)+\frac{1}{2}\int_{0}^{T}\hat{W}_{s}d\tilde{P}^{\hat{W}}(s,0)-T_{T}^{\hat{W}}\nonumber\\&=& \lim_{k\rightarrow\infty}[\int_{0}^{T}Z^1_{n_{k},s}dM(s,0)+\frac{1}{2}\int_{0}^{T}Z^1_{n_{k},s}d\tilde{P}^{Z^1_{n_{k}}}(s,0)- T_{T}^{Z^1_{n_{k}}}]\nonumber\\&=& g \quad \mbox{a.s.}
\end{eqnarray}
The fact that $\int_{0}^{T}Z^1_{n_{k},s}dM(s,0)+\frac{1}{2}\int_{0}^{T}Z^1_{n_{k},s}d\tilde{P}^{Z^1_{n_{k}}}(s,0)- T_{T}^{Z^1_{n_{k}}}$ converges a.s. to $g$ as $k\rightarrow \infty$ implies that this limit equals the a.s. limit of $g_{n}$.    

Obviously, $\int_{0}^{T}\hat{W}_{s}dM(s,0)+\frac{1}{2}\int_{0}^{T}\hat{W}_{s}d\tilde{P}^{\hat{W}}(s,0)-T_{T}^{\hat{W}}$ belongs to the space $\mathcal{K}_{0}^{r}$. This ends the proof.
\end{proof}
\begin{corollary}\label{mjaft1}
If we restrict the set $\mathcal{K}^{r}_{0}$ in Theorem \ref{themu} to continuous of bounded variation $\alpha$-admissible trading strategies $Z^1$, the set $\mathcal{C}^{r}_{0}$ becomes a convex cone. Moreover, it is Fatou-closed and hence by  a simple application of the Krein-Smulian Theorem and the dominated convergence theorem, $\mathcal{C}^{r}_{0}\cap L^{\infty}$ is $\sigma(L^{\infty},L^{1})$-closed. The fact that $\mathcal{C}^{r}_{0}$ is Fatou-closed can be proven by contradiction. Indeed, suppose $\mathcal{C}^{r}_{0}$ restricted to continuous and bounded variation trading strategies is not Fatou-closed. This means that there exists a sequence $(V_{T}^{Z^1_{n}})_{n\geq 1}$ with $V_{T}^{Z^1_{n}}\geq -\alpha$, $\alpha>0$, and $\lim_{n\rightarrow \infty}V_{T}^{Z^1_{n}}=f_{0}$ a.s. such that $f_{0}$ is not in $\mathcal{C}^{r}_{0}$. However, as in Lemma \ref{lemmanew}, we can use the results obtained by Bank \& Baum (2004) and \c{C}etin et.al (2004) to approximate  u.c.p. stochastic integrals of the form $\int_{0}^{T}Z^1_{s}dM(s,0)+\frac{1}{2}\int_{0}^{T}Z^1_{s}d\tilde{P}^{Z^1}(s,0)$ with stochastic integrals of continuous and bounded variation integrands. This means we can find a sequence of $\alpha$-admissible trading strategies $(K^{m}_{n})_{m\geq 1}$, continuous and of bounded variation in the implicit transaction cost economy, such that $V_{T}^{K^{m}_{n}}$ tends to $V_{T}^{Z^1_{n}}$ as $m\rightarrow \infty$ for any $n\geq 1$. It follows that we can find a subsequence of  $V_{T}^{K^{m}_{n}}$ tending a.s. to $V_{T}^{Z^1_{n}}$. Therefore, as $n\rightarrow \infty$, $\lim_{m\rightarrow \infty}V_{T}^{K^{m}_{n}}$ will be equal to the limit of $V_{T}^{Z^1_{n}}$. Using the proof in Theorem \ref{themu}, this limit will be necessarily of the form $\int_{0}^{T}F_{s}dM(s,0)+\frac{1}{2}\int_{0}^{T}F_{s}d\tilde{P}^{F}(s,0)$, where $F$ is a predictable process. By assumption $f_0$ is not in $\mathcal{C}^{r}_{0}$ restricted to continuous of bounded variation trading strategies. This implies that $F$ is a non-continuous process, and this contradicts the fact that $\mathcal{C}^{r}_{0}$ is Fatou-closed in the implicit transaction cost economy. 
\end{corollary}
Let us prove another simple lemma before proving the FFTAP under implicit transaction costs. 
\begin{lemma}\label{lemmnflvr}
Suppose there exists an equivalent $\mathbb{Q}$-local martingale measure for $M(\cdot,0)$. Then the standard economy in \ref{app1} satisfies the NFLVR property. As a result, the implicit transaction cost economy has NFLVR. 
\end{lemma} 
\begin{proof}
Assume there exists an equivalent $\mathbb{Q}$-local martingale measure. Theorem \ref{std} in \ref{app1} shows the standard economy has NFLVR. Then, it can be easily deduced that NFLVR in the standard economy implies NFLVR in the implicit transaction cost economy. In fact, $V_{T}^{Z^1}\leq \int_{0}^{T}Z^1_{s}dM(s,0)$ a.s. for every $\alpha$-admissible trading strategy $Z^1$ in the implicit transaction cost economy. Moreover, from $V_{T}^{Z^1}\geq-\alpha$ it follows that $\int_{0}^{T}Z^1_{s}dM(s,0)\geq-\alpha$. That is, every $Z^1$ $\alpha$-admissible trading strategy  in the implicit transaction cost economy is also  $\alpha$-admissible in the standard economy. The assertion then follows.
\end{proof}
\begin{theorem}\label{ftap}
The implicit transaction cost economy satisfies NFLVR if and only if there exists an equivalent $\mathbb{Q}$-local martingale measure to $\mathbb{P}$ such that $M(\cdot,0)$ is a $\mathbb{Q}$-local martingale. 
\end{theorem}
\begin{proof}
Suppose there exists  an equivalent $\mathbb{Q}$-local martingale measure. Using Lemma \ref{lemmnflvr}, we have NFLVR in the standard economy. More generally, for every function $h\in\mathcal{C}_{s}$, $\mathbb{E}_{\mathbb{P}}(h)\leq0$, where $\mathcal{C}_{0}^{s}=(\mathcal{K}_{0}^{s}-L_{+}^{0})\cap L^{\infty}$ and $\mathcal{K}_{0}^{s}=\{\int_{0}^{T}Z^1_{s}dM(s,0): \text{$Z^1$  is $\alpha$-admissible and predictable}\}$. This of course is true for the closure of $\mathcal{C}_{0}^{s}$ with respect to the norm topology of $L^{\infty}$, and we have that $\mathcal{\overline{C}}_{s}\cap L_{+}^{\infty}=\{0\}$. Invoking again Lemma \ref{lemmnflvr}, we conclude that the implicit transaction cost economy has NFVLR. In view of this result,  $\mathcal{\overline{C}}^{r}_{0}\cap L_{+}^{\infty}=\{0\}$, where $\mathcal{\overline{C}}^{r}_{0}$ is the closure set of $\mathcal{C}^{r}_{0}\cap L^{\infty}$ with respect to the norm topology of $L^{\infty}$. Note that the last equality is coherent with Definition \ref{def4}. 

For the converse, assume the NFLVR holds. Since $M(\cdot,0)$ satisfies the NFLVR property, $(\mathcal{C}^{r}_{0}\cap L^{\infty})\cap L_{+}^{\infty}=\{0\}$. 

Assume for a moment that $M(\cdot,0)$ is bounded. Next, define for each $u<t$, $A_{u}\in\mathcal{F}_{u}$, $c\in\mathbb{R}$ the elementary process $\psi^{c}(s)=c\textbf{1}_{A_{u}}\textbf{1}_{(u,t]}(s)$, which is  predictable, c\'agl\'ad with finitely many jumps, and trivially of bounded quadratic variation. The  portfolio value of the strategy $\psi^{c}$ is then given by
\begin{equation}
V_{T}^{\psi^{c}}=c\textbf{1}_{A_{u}}(M(t,0)-M(u,0))+\frac{1}{2}c\textbf{1}_{A_{u}}(\tilde{P}^{\psi^{c}}(t,0)-\tilde{P}^{\psi^{c}}(u,0))-T_{T}^{\psi^{c}}
\end{equation}
By assumption, $T_{T}^{\psi^{c}}$ and $M(\cdot,0)$, thus also $\tilde{P}^{\psi^{c}}(\cdot,0)$, are bounded, thus $V_{T}^{\psi^{c}}$ is bounded and  belongs to $\mathcal{C}^{r}_{0}\cap L^{\infty}$. 

As in Lemma \ref{lemmanew}, stochastic integrals of the form $\int_{0}^{T}Z^1_{s}dM(s,0)+\frac{1}{2}\int_{0}^{T}Z^1_{s}d\tilde{P}^{Z^1}(s,0)$, with predictable integrands $Z^1$, can be approximated u.c.p. by stochastic integrals with continuous and bounded variation integrands. Consequently, there exists a sequence of continuous and bounded variation  strategies $\psi^{c}_{n}$ such that $\int_{0}^{T}\psi^{c}_{n,s}dM(s,0)+\frac{1}{2}\int_{0}^{T}\psi^{c}_{n,s}d\tilde{P}^{\psi^{c}_{n}}(s,0)$  tends u.c.p to $\int_{0}^{T}\psi^{c}_{s}dM(s,0)+\frac{1}{2}\int_{0}^{T}\psi^{c}_{s}d\tilde{P}^{\psi^{c}}(s,0)$.
One can easily see this by taking trading strategies of the following form 
\begin{equation}
\psi^{c}_{n,s} =n\int_{s-\frac{1}{n}}^{s}\psi^{c}_{u}du
\end{equation}
which are continuous of bounded variation and $\psi^{c}_{u}=\psi^{c}_{0}$ for $u<0$. Then, $\psi^{c}_{n}$ tends a.s. to $\psi^{c}$, and $\psi^{c}_{n}$ are uniformly bounded. The result then follows by the dominated convergence theorem for stochastic integrals.  
  
Using Corollary \ref{mjaft1}, $\mathcal{C}^{r}_{0}$ restricted to continuous $\alpha$-admissible processes of bounded variation is a convex cone, Fatou-closed, and $\mathcal{C}^{r}_{0}\cap L^{\infty}$ is $\sigma(L^{\infty},L^1)$-closed. Then, Kreps-Yan Separation Theorem (see Kreps (1981) and Yan (1980)) implies that there exists a probability  $\mathbb{Q}$ equivalent to $\mathbb{P}$  such that $\mathbb{E}_{\mathbb{Q}}(h)\leq0$ for every $h\in\mathcal{C}^{r}_{0}\cap L^{\infty}$.  

A look at $V_{T}^{\psi^{c}_{n}}$ shows that $V_{T}^{\psi^{c}_{n}}$ are in $\mathcal{C}^{r}_{0}\cap L^{\infty}$ restricted to continuous $\alpha$-admissible processes of bounded variation, $T_{T}^{\psi^{c}_{n}}=0$ for all $\psi^{c}_{n}$. 
   
Another application of the ordinary dominated convergence theorem yields
\begin{eqnarray}
&&\lim_{n\rightarrow\infty}\mathbb{E}_{\mathbb{Q}}(\int_{0}^{T}\psi^{c}_{n,s}dM(s,0)+\frac{1}{2}\int_{0}^{T}\psi^{c}_{n,s}d\tilde{P}^{\psi^{c}_{n}}(s,0))\nonumber\\&&=\mathbb{E}_{\mathbb{Q}}(\int_{0}^{T}\psi^{c}_{s}dM(s,0)+\frac{1}{2}\int_{0}^{T}\psi^{c}_{s}d\tilde{P}^{\psi^{c}}(s,0))\nonumber\\&&=\mathbb{E}_{\mathbb{Q}}(c\textbf{1}_{A_{u}}[(M(t,0)+\frac{1}{2}\tilde{P}^{\psi^{c}}(t,0))-(M(u,0)+\frac{1}{2}\tilde{P}^{\psi^{c}}(u,0))]
\end{eqnarray}
By previous considerations we have that $\mathbb{E}_{\mathbb{Q}}(V_{T}^{\psi^{c}_{n}})\leq0$ for all $n$ and so 
\begin{eqnarray}
\lim_{n\rightarrow\infty}\mathbb{E}_{\mathbb{Q}}(V_{T}^{\psi^{c}_{n}})&=&\mathbb{E}_{\mathbb{Q}}(c\textbf{1}_{A_{u}}[(M(t,0)+\frac{1}{2}\tilde{P}^{\psi^{c}}(t,0))-(M(u,0)+\frac{1}{2}\tilde{P}^{\psi^{c}}(u,0))])\nonumber\\&=&\mathbb{E}_{\mathbb{Q}}(c\textbf{1}_{A_{u}}[A(t,0)-A(u,0)])\leq0
\end{eqnarray} 
and
\begin{eqnarray}
\lim_{n\rightarrow\infty}\mathbb{E}_{\mathbb{Q}}(V_{T}^{\psi^{-c}_{n}})&=&\mathbb{E}_{\mathbb{Q}}(-c\textbf{1}_{A_{u}}[(M(t,0)+\frac{1}{2}\tilde{P}^{\psi^{c}}(t,0))-(M(u,0)+\frac{1}{2}\tilde{P}^{\psi^{c}}(u,0))])\nonumber\\&=&\mathbb{E}_{\mathbb{Q}}(-c\textbf{1}_{A_{u}}[A(t,0)-A(u,0)])\leq0
\end{eqnarray} 
for $c>0$. 

By the same argument, $c<0$ implies that $\mathbb{E}_{\mathbb{Q}}[c\textbf{1}_{A_{u}}(B(t,0)-B(u,0))]=0$. 
Equipped with the above facts, one can easily note that $\mathbb{E}_{\mathbb{Q}}[c\textbf{1}_{A_{u}}(M(t,0)-M(u,0))]=0$, thus $M(\cdot,0)$ is a $\mathbb{Q}$-martingale. 

For the general case, we may apply Corollary 1.2 in Delbaen \& Schachermayer (1994) to show that $\mathbb{Q}$ is a local martingale measure for $M(\cdot,0)$. This concludes the proof of the theorem.
\end{proof}
\begin{remark}
Note that similarly as in the Appendix A.4 in the paper of \c{C}etin \textit{et al.} (2004) one can prove a num\'eraire invariance theorem for the implicit transaction cost economy. More importantly, $M(\cdot,0)$ being a $\mathbb{Q}$-local martingale implies that also the economy with only implicit transaction costs as measured by the bid-ask spread and with continuous of bounded variation trading strategies (Guasoni \textit{et al.} (2010)) is free of arbitrage opportunities. 
\end{remark}
\section{Hedging in the implicit transaction cost economy}\label{sec2}  
In this section we  assume that an equivalent $\mathbb{Q}$-local martingale measure exists and hence both the standard and the implicit transaction cost economy in the previous section are free of arbitrage opportunities.  The arbitrage-free property is then used to determine whether or not the implicit transaction cost economy is complete. In addition, for this section, we suppose that $M(\cdot,0)$ is a special semimartingale (see Protter (2004) for a definition) belonging to the space $H^2$, with finite norm under $\mathbb{Q}$ given by $||[M,M]_{\infty}^{1/2}||_{L^2}+||\int_{0}^{\infty}|dA_{s}|||_{L^2}<\infty$.  
\begin{definition}\label{cc}
\begin{itemize}
\item [(i)] A contingent claim is any $\mathcal{F}_{T}$-measurable random variable $X$ with $E_{\mathbb{Q}}(X^2)<\infty$. A contingent claim $X$ is called hedgeable (or replicable or attainable) if there exists an admissible self-financing portfolio $Z=(Z^0, Z^1)$  such that $ \mathbb{Q}$ a.s.
\begin{eqnarray}\label{cc}
X&=&V_{T}^{Z^1}
\end{eqnarray}
\item[(ii)] The implicit transaction cost economy is called $\mathcal{F}_{T}$-complete if every contingent claim is hedgeable. 
\end{itemize}
\end{definition}
We also assume that the equivalent $\mathbb{Q}$-local martingale measure is unique, so that the  standard economy satisfies the SFTAP. That is, the uniqueness of the equivalent  local martingale measure $\mathbb{Q}$  implies the standard economy is complete.

In other words, $X=b + \int_{0}^{T}Z^1_{s}dM(s,0)$ for some admissible predictable trading strategies $Z^1$ such that $E_{\mathbb{Q}}(\int_{0}^{T}(Z^1_{s})^{2}d[M(s,0),M(s,0)]_{s})<\infty$. Note that $E_{\mathbb{Q}}(X)=b$ gives the risk-neutral price of the  contingent claim $X$ in the standard economy. 

Combining this result, for each $X$, there exists an admissible predictable trading strategy $D$, in the standard economy, with $E_{\mathbb{Q}}(\int_{0}^{T}$ $(D_{s})^{2}d[M(s,0),M(s,0)]_{s})^{1/2}<\infty$ such that Equation (\ref{cc}) can be written as
\begin{equation}
X=b +\int_{0}^{T}D_{s}dM(s,0)=V_{T}^{Z^1}
\end{equation}
with $b=Z^{0}_{0}+D_{0}M(0,0)=E_{\mathbb{Q}}(X)$. From Lemma 4.1 in \c{C}etin \textit{et al.} (2004), $\int_{0}^{T}D_{s}dM(s,0)$ can be approximated in $L^2(\mathbb{Q})$-space by $\int_{0}^{T}D_{n,s}dM(s,0)$, with $D_{n}$ being admissible, continuous of bounded variation trading strategies with $E_{\mathbb{Q}}(\int_{0}^{T}(D_{n,s})^{2}d[M(s,0),M(s,0)]_{s})<\infty$. Now put $Z^{0}_{n, 0}=E_{\mathbb{Q}}(X)$ and $D_{n,0}=0$ for all $n$. Then, $Z^{0}_{n, 0}+D_{n,0}M(0,0)+\int_{0}^{T}D_{n,s}dM(s,0)+\frac{1}{2}\int_{0}^{T}D_{n,s}d\tilde{P}^{D_{n}}(s,0)-T^{D_{n}}_{T}$ converges in $L^2(\mathbb{Q})$ to $E_{\mathbb{Q}}(X)+\int_{0}^{T}D_{s}dM(s,0)+\frac{1}{2}\int_{0}^{T}D_{s}d\tilde{P}^{D}(s,0)$ as $n\rightarrow \infty$. 

We can thus come to the conclusion that the implicit transaction cost economy is incomplete in the sense of Definition \ref{cc}, but at least the investor could obtain a better approximation of the contingent claim's value (in the $L^2(\mathbb{Q})$-sense) by using continuous of bounded variation strategies. In a few words, it says that the use of continuous and of bounded variation trading strategies is the best choice to make when dealing with the problem of hedging in markets with implicit transaction costs.
\section{Examples of implicit transaction cost economies}\label{bs}
The goal of this section is to apply the results from the previous sections to an implicit transaction cost economy that is linear and non-linear in the order size. 
\subsection{Linear implicit transaction cost economy}
The linear form of the implicit transaction cost economy is motivated by the work of Engle \& Patton (2004) and Hasbrouck (1991), which we discussed in Sec. \ref{the model}. Other studies also find a linear relationship between the trade price and the traded volume. For example, Blais (2005) and Blais \& Protter (2010) show by using a linear regression model that for liquid stocks this relationship is linear with  time varying slope and intercept. The proposed linear supply curve is of the form
\begin{equation}
S(t,y)=N_ty+S(t,0)
\end{equation}
where $N$ is a stochastic process with continuous paths and $S(\cdot,0)$ the marginal price process. Note that the sensitivity of the price  to the order size, $N$, is the same both for the buy and  sell orders. Blais (2005) shows that the hypothesis $N_t=0$ can be rejected with a significance level of $0.9999$.
 
Going back to our framework, one trivially notes that we must define two processes for the marginal ask and the bid price, i.e. $A(\cdot,0)$ and $B(\cdot,0)$.  We assume that these processes are of the following form
\begin{eqnarray}\label{ask1}
A(t,0)&=&M(t,0)+\gamma_t\\B(t,0)&=&M(t,0)-\gamma_t
\end{eqnarray}
where  $\gamma_t>0$ is a continuous stochastic process. Note that $A(t,0)>B(t,0)$ and that $P(t,0)=2\gamma_t$.  

Using the equations above, we define $M(t, y)$ and $P(t, y)$ in the implicit transaction cost economy as follows 
 \begin{eqnarray}
M(t,y) &=& \frac{1}{2}[(\beta_t+\lambda_t)y + 2M(t,0)]=\frac{1}{2}[(\beta_t+\lambda_t)y] + M(t,0)\\
P(t,y) &=& A(t,y)- B(t,y)=(\beta_t-\lambda_t)y+P(t,0)=(\beta_t-\lambda_t)y+2\gamma_t
\end{eqnarray}
where $\beta_t>0$, $\lambda_t>0$, with $\beta_t\geq \lambda_t$, are again continuous stochastic processes.

We can use $M(t,y)$ and $P(t,y)$ to derive the ask and the bid price, which are respectively given by
\begin{eqnarray}\label{ask2}
A(t,y)&=&A(t,0)+\beta_ty\\B(t,y)&=&B(t,0)+\lambda_ty
\end{eqnarray}
or 
\begin{eqnarray}\label{ask}
A(t,y)&=&M(t,0)+\gamma_t+\beta_ty\\B(t,y)&=&M(t,0)-\gamma_t+\lambda_ty\label{askik}
\end{eqnarray}
Eqs. (\ref{ask}) and (\ref{askik}) may be negative from a mathematical point of view, but this is practically impossible since the ask and the bid price are positive. These equations should be read as the supply curve equations of the ask and the bid price, which means that they give the price of a share of the stock when the investor wants to buy or to sell $|y|$ shares of a given stock. 

Fix now $t$. The coefficients  $\beta_t$ and $\lambda_t$ determine how the ask and the bid price respond to a change in the order size. One would expect these coefficients to be very small for liquid stocks. Then,  (\ref{ask}) and (\ref{askik}) give the equations of two lines, with positive intercepts and slopes. It thus follows that the ask and the bid prices are an increasing  function of the order size $y$. Another way of saying this is that the ask and the bid price are high for  high positive orders and low for high negative orders. 

One can easily check that the linear implicit transaction cost economy satisfies the assumptions made in the previous sections, and so the value of the portfolio for every admissible trading strategy can be derived as in Subsec. \ref{cont}. 
\subsection{Non-linear implicit transaction cost economy}\label{european}
In the  same spirit of \c{C}etin \textit{et. al} (2002) and \c{C}etin \textit{et. al} (2004)  we consider an extension of the BS economy that includes implicit transaction costs. To this purpose, we suppose that as in the linear case the marginal prices are governed by the following equations
\begin{eqnarray}\label{ec}
A(t,0)&=&M(t,0)+\gamma_t\\
B(t,0)&=&M(t,0)-\gamma_t
\end{eqnarray}
with $\gamma_t>0$, and that $M(t,y)$ and $P(t,y)$  are given as
\begin{eqnarray}\label{ec1}
M(t,y)&=&e^{\alpha y}M(t,0)\\
P(t,y)&=&2e^{\alpha y}\gamma_t\label{Ec}
\end{eqnarray}
with $\alpha>0$.  Given $M(t,y)$ and $P(t,y)$, we readily compute that the prices in the implicit transaction cost economy have the following expressions
\begin{eqnarray}\label{pa}
A(t,y)&=&e^{\alpha y}M(t,0)+e^{\alpha y}\gamma_t  \\\label{pa1}
B(t,y)&=&e^{\alpha y}M(t,0)-e^{\alpha y}\gamma_t
\end{eqnarray}
\subsection{Black-Scholes  model}
Suppose $M(\cdot,0)$ is a semimartingale similar to that used by BS. In the BS model the stock price evolves according to a geometric Brownian motion with constant drift and volatility. It follows that, under the probability measure $\mathbb{P}$, the stock price dynamic is given by
\begin{equation}
dM(t,0)=\mu M(t,0)dt + \sigma M(t,0)dW_{t}
\end{equation}
where $W$ is a standard Brownian motion zero at $t=0$ and $\mu$, $\sigma$ are constants. 

As is well-known the BS model is free of arbitrage opportunities and complete, so that  every contingent claim is hedgeable. In particular, this means that the theory presented in the previous sections is well adapted to the BS model, and thus the stock price in the BS model is an ideal candidate for playing the role of $M(\cdot, 0)$. 
\subsection{European call option under transaction costs} 
Consider an European call option with maturity date $T$ and strike price $K$. We suppose the European call option contains a physical delivery feature. This means that in order to avoid implicit transaction costs, the holder of the call would buy the stock at price $K$, and sell it at the highest price. The payoff at time $T$ of the call is thus supposed to be $X=\mbox{max}[M(T,0)-K, 0]$. Note that by the Eqs. (\ref{ask}) and (\ref{askik}), (\ref{pa}) and (\ref{pa1}), $B(T,0)$ is greater than $B(T,y)$ for every $y<0$ which in turn is smaller than $M(T,0)$.

Under the BS model, the price of the European call option at time $t$ is given by
\begin{eqnarray}
\phi(t,M(t,0))&=&M(t,0)\Phi\left(\frac{1}{\sigma\sqrt{T-t}}(\log{\frac{M(t,0)}{K}}+\frac{1}{2}\sigma^2(T-t))\right)\nonumber\\&-&K\Phi\left(\frac{1}{\sigma\sqrt{T-t}}(\log{\frac{M(t,0)}{K}}-\frac{1}{2}\sigma^2(T-t))\right)
\end{eqnarray}
where $\Phi(\cdot)$ is the standard cumulative normal distribution function, and
\begin{equation}\label{imp}
\Phi(d_t) = \Phi\left(\frac{1}{\sigma\sqrt{T-t}}(\log{\frac{M(t,0)}{K}}+\frac{1}{2}\sigma^2(T-t))\right)
\end{equation}
gives the delta of the European call option or, equivalently, the replicating trading strategy of the European call option in the BS economy. 

This trading strategy is non-negative, continuous of unbounded variation, and has finite quadratic variation. However, this strategy is allowed in the model illustrated previosly. Indeed, plugging this trading strategy in Eq. (\ref{zari}) or (\ref{zari1}) and supposing that the implicit transaction cost economy is linear with $M(\cdot,0)$ as in the BS model, the value of the portfolio amounts to
\begin{eqnarray}\label{zari2}
V_T^{\Phi}&=&Z^0_{0} + \Phi(d_0) M(0,0)+ \int_{0}^{T}\Phi(d_s)dM(s,0)+\frac{1}{2}\int_{0}^{T}\Phi(d_s)d\tilde{P}^{\Phi}(s,0) \nonumber\\&-& \frac{1}{2}\int_{0}^{T}(\beta_s+\lambda_s)ds-\frac{1}{2}\int_{0}^{T}(\beta_s-\lambda_s)\mbox{sgn}(\Delta \Phi(d_s))ds
\end{eqnarray}
with
\begin{equation}
T_{T}^{\Phi}=\frac{1}{2}\int_{0}^{T}(\beta_s+\lambda_s)ds+\frac{1}{2}\int_{0}^{T}(\beta_s-\lambda_s)\mbox{sgn}(\Delta \Phi(d_s))ds
\end{equation}
The BS economy being complete implies that 
\begin{eqnarray}\label{1}
V_T^{\Phi}&=&E_{\mathbb{Q}}(X)+ \int_{0}^{T}\Phi(d_s)dM(s,0) +\frac{1}{2}\int_{0}^{T}\Phi(d_s)d\tilde{P}^{\Phi}(s,0)\nonumber\\&-&\frac{1}{2}\int_{0}^{T}(\beta_s+\lambda_s)ds-\frac{1}{2}\int_{0}^{T}(\beta_s-\lambda_s)\mbox{sgn}(\Delta \Phi(d_s))ds\nonumber\\&=&X+\frac{1}{2}\int_{0}^{T}\Phi(d_s)d\tilde{P}^{\Phi}(s,0)- \frac{1}{2}\int_{0}^{T}(\beta_s+\lambda_s)ds-\frac{1}{2}\int_{0}^{T}(\beta_s-\lambda_s)\mbox{sgn}(\Delta \Phi(d_s))ds\nonumber\\
\end{eqnarray}
The presence of the transaction costs, as one can see from Eq. (\ref{1}), have a considerable effect on $V_T^{\Phi}$. Also, increases in $\beta_s$ and $\lambda_s$ produce decreases in the value of the portfolio.  

However, referring to Sec. \ref{sec2} there exists a sequence $\Phi_{n}$ of admissible continuous bounded variation trading strategies such that the impact of price on the portfolio's value can be removed, in an $L^2(\mathbb{Q})$-sense. 

The results for the European put option case are completely analogous.

Under the non-linear implicit transaction cost economy and $M(\cdot,0)$ as in the BS model, the replicating trading strategy of the European call of Subsec. \ref{european} is the same and, the value of the portfolio is given by
\begin{eqnarray}\label{111}
V_T^{\Phi}&=&E_{\mathbb{Q}}(X)+ \int_{0}^{T}\Phi(d_s)dM(s,0) +\frac{1}{2}\int_{0}^{T}\Phi(d_s)d\tilde{P}^{\Phi}(s,0)\nonumber\\&-&\frac{1}{2}\int_{0}^{T}\alpha M(s,0)ds-\int_{0}^{T}\alpha\gamma_{s}\mbox{sgn}(\Delta \Phi(d_s))ds\nonumber\\&=&X+\frac{1}{2}\int_{0}^{T}\Phi(d_s)d\tilde{P}^{\Phi}(s,0)- \frac{1}{2}\int_{0}^{T}\alpha M(s,0)ds\nonumber\\&-&\int_{0}^{T}\alpha\gamma_{s}\mbox{sgn}(\Delta \Phi(d_s))ds
\end{eqnarray}
\section{Conclusion}\label{conc}
The goal of this paper has been the construction of an arbitrage pricing model with implicit transaction costs. To facilitate the understanding of the model, the paper illustrates the model by using a linear and a non-linear implicit transaction cost economy. 

With respect to the existing literature, the paper contributes by incorporating general implicit transaction costs into the standard arbitrage pricing theory. This is achieved by letting the ask and the bid price to depend on the traded volume. And, more importantly, it is proved that the first fundamental theorem of asset pricing (FFTAP) as stated in the standard arbitrage pricing theory is still valid. In addition, contrary to the financial literature, the proposed model shows that continuous trading is not only restricted to trading strategies of bounded variation, but also to trading strategies of bounded quadratic variation. Last but not least, the proposed model shows that the use of continuous and bounded variation trading strategies removes price impact resulting from the trading activities. In other words, investors afford only the traditional implicit transaction costs given by half the difference between the ask and the bid price when they trade with these kinds of trading strategies.

The self-financing portfolios differ from those proposed in \c{C}etin \textit{et al.} (2004) by the fact that prices are different for a buy  and a sell order, and are different from those of Guasoni \textit{et al.} (2010) given that prices depend on the order size. Furthermore, when the trading strategies are continuous and of bounded variation the self-financing portfolios are similar to those of Guasoni \textit{et al.} (2010). It follows that the FFTAP  restricted to these trading strategies holds also for the model in Guasoni \textit{et al.} (2010). 

Future directions for research may include: a) proving the FFTAP under more general trading strategies, here the FFTAP is proved for c\'agl\'ad (left-continuous with right limits) adapted of bounded quadratic variation trading strategies with finitely many jumps, b) developing the continuous version of the self-financing portfolio for general forms of the marginal mid-price and marginal bid-ask spread, here we assume that these are c\'adl\'ag locally bounded semimartingales, c) studying in greater details the problem of hedging under implicit transaction costs. 
\appendix
\section{Self-financing portfolio}\label{app}
The aim of this appendix is to derive the continuous version of Eq. (\ref{port}) for every $t\in[0,T]$. Let  $\sigma_{n}: 0=\tau_{0}^n\leq\tau_{1}^n\leq...\leq \tau_{i_n}^n =t$ be a sequence of random partitions tending to the identity. Take the first and the third sum of Eq. (\ref{port}), and evaluate these as $n\rightarrow \infty$. As it is easily seen, these are just It\'o integrals. The second  result follows by Definition \ref{tradin} and Assumption \ref{mid}. 

With regards to the other terms, the proof follows the pattern of that given for Theorem A.3 in \c{C}etin \textit{et al.} (2004). Before we begin this proof, let us recall an important lemma which in our case holds for all jumps of $Z^1$ since by assumption $Z^1$ has finitely many jumps. 
\begin{lemma}\label{lem}
Let $f$ be a c\'adl\'ag (c\'agl\'ad) function on $[a,b]$. Define  $J^{f}_{\epsilon}=\{y\in[a,b]: |f(y+) - f(y-)|>\epsilon\}$. Then, $J^{f}_{\epsilon}$ is finite for every $\epsilon>0$. 
\end{lemma} 
\begin{proof}
See, for example,  Folland (1999).
\end{proof}
For the sake of simplicity and to save space, we will often denote  
\begin{equation}
M(\tau_{i}^{n},\Delta Z^1_{\tau_{i}^{n}})+\frac{1}{2}P(\tau_{i}^{n},\Delta Z^1_{\tau_{i}^{n}})\mbox{sgn}(\Delta Z^1_{\tau_{i}^{n}})\Delta Z^1_{\tau_{i}^{n}}
\end{equation}
by $W(\tau_{i}^{n},\Delta Z^1_{\tau_{i}^{n}})$. 

Write now the limit of $\sum_{i=1}^{n}[W(\tau_{i},\Delta Z^1_{\tau_{i}})-W(\tau_{i},0)]\Delta Z^1_{\tau_{i}}$ as $\lim_{n\rightarrow \infty}\sum_{i\geq1}$ $[W(\tau_{i}^{n},\Delta Z^1_{\tau_{i}^{n}})-W(\tau_{i}^{n},0)]\Delta Z^1_{\tau_{i}^{n}}$.

Then, one can easily verify that this limit is also equal to 
\begin{eqnarray}
&&\lim_{n\rightarrow \infty}\sum_{i\geq1}[W(\tau_{i}^{n},\Delta Z^1_{\tau_{i}^{n}})-W(\tau_{i}^{n},0)]\Delta Z^1_{\tau_{i}^{n}}\nonumber\\&&=\lim_{n\to\infty}\sum_{|Z^1_{\tau_{i}^{n}} - Z^1_{\tau_{i-1}^{n}}|\leq \epsilon}[W(\tau_{i}^{n},\Delta Z^1_{\tau_{i}^{n}})-W(\tau_{i}^{n},0)]\Delta Z^1_{\tau_{i}^{n}}\nonumber\\&&+\lim_{n\to\infty}\sum_{|Z^1_{\tau_{i}^{n}} - Z^1_{\tau_{i-1}^{n}}|> \epsilon}[W(\tau_{i}^{n},\Delta Z^1_{\tau_{i}^{n}})-W(\tau_{i}^{n},0)]\Delta Z^1_{\tau_{i}^{n}}
\end{eqnarray}
Therefore, we have to prove that the terms on the right-hand side exist as $n\rightarrow\infty$. Let's focus on the second term and suppose $Z^1$ is a c\'adl\'ag process. By the above lemma, $J^{Z^1}_{\epsilon}$ is finite and therefore the second limit can be written as
\begin{equation}\label{neweq}
\sum_{s\in J^{Z^1}_{\epsilon}}[W(s,Z^1_{s}-Z^1_{s-})-W(s,0)](Z^1_{s}-Z^1_{s-})
\end{equation}
As regards with the first term we apply Taylor approximation to each $M(\tau_{i}^{n},\cdot)$ and $P(\tau_{i}^{n},\cdot)$. Then,
\begin{eqnarray}
&&\lim_{n\to\infty}\sum_{|Z^1_{\tau_{i}^{n}} - Z^1_{\tau_{i-1}^{n}}|\leq \epsilon}[(M(\tau_{i}^{n},\Delta Z^1_{\tau_{i}^{n}})-M(\tau_{i}^{n},0))\nonumber\\&&+\frac{1}{2}(P(\tau_{i}^{n},\Delta Z^1_{\tau_{i}^{n}})\mbox{sgn}(\Delta Z^1_{\tau_{i}^{n}})-P(\tau_{i}^{n},0)\mbox{sgn}(\Delta Z^1_{\tau_{i}^{n}}))]\Delta Z^1_{\tau_{i}^{n}}\nonumber\\&&=\lim_{n\to\infty}\sum_{|Z^1_{\tau_{i}^{n}} - Z^1_{\tau_{i-1}^{n}}|\leq \epsilon} [M'(\tau_{i}^{n},0)+\frac{1}{2}P'(\tau_{i}^{n},0)\mbox{sgn}(\Delta Z^1_{\tau_{i}^{n}})](\Delta Z^1_{\tau_{i}^{n}})^{2}\nonumber\\&&+\lim_{n\to\infty}\sum_{|Z^1_{\tau_{i}^{n}} - Z^1_{\tau_{i-1}^{n}}|\leq \epsilon} [o^{M}(\tau_{i}^{n},\Delta Z^1_{\tau_{i}^{n}})+\frac{1}{2}o^{P}(\tau_{i}^{n},\Delta Z^1_{\tau_{i}^{n}})\mbox{sgn}(\Delta Z^1_{\tau_{i}^{n}})]\Delta Z^1_{\tau_{i}^{n}}\nonumber\\&&=\lim_{n\to\infty}\sum_{i\geq1}[M'(\tau_{i-1}^{n},0)+\frac{1}{2}P'(\tau_{i-1}^{n},0)\mbox{sgn}(\Delta Z^1_{\tau_{i}^{n}})](\Delta Z^1_{\tau_{i}^{n}})^{2}\nonumber\\&&+\lim_{n\to\infty}\sum_{i\geq1}[(M'(\tau_{i}^{n},0)+\frac{1}{2}P'(\tau_{i}^{n},0)\mbox{sgn}(\Delta Z^1_{\tau_{i}^{n}}))\nonumber\\&&-(M'(\tau_{i-1}^{n},0)+\frac{1}{2}P'(\tau_{i-1}^{n},0)\mbox{sgn}(\Delta Z^1_{\tau_{i}^{n}}))](\Delta Z^1_{\tau_{i}})^{2}\nonumber\\&&-\lim_{n\to\infty}\sum_{|Z^1_{\tau_{i}^{n}} - Z^1_{\tau_{i-1}^{n}}|> \epsilon}[M'(\tau_{i}^{n},0)+\frac{1}{2}P'(\tau_{i}^{n},0)\mbox{sgn}(\Delta Z^1_{\tau_{i}^{n}})](\Delta Z^1_{\tau_{i}^{n}})^{2}\nonumber\\&&+\lim_{n\to\infty}\sum_{|Z^1_{\tau_{i}^{n}} - Z^1_{\tau_{i-1}^{n}}|\leq \epsilon}[o^{M}(\tau_{i}^{n},\Delta Z^1_{\tau_{i}^{n}})+\frac{1}{2}o^{P}(\tau_{i}^{n},\Delta Z^1_{\tau_{i}^{n}})\mbox{sgn}(\Delta Z^1_{\tau_{i}^{n}})]\Delta Z^1_{\tau_{i}^{n}}\nonumber\\
\end{eqnarray}
Let's focus on the first limit. This limit converges to $\int_{0}^{t}[M'(s,0)+\frac{1}{2}P'(s,0)\mbox{sgn}(\Delta Z^1_{s})]$ $d[Z^1,Z^1]_{s}$. To see this, note that $Z^1$ has bounded quadratic variation, i.e.
\begin{equation}
\lim_{n\to\infty}  \sum_{i\geq 1} (Z^1_{\tau_{i}^{n}}-Z^1_{\tau_{i-1}^{n}})^{2}<\infty
\end{equation}
Therefore, we may approach Lebesgue-Stieltjes integral to show that the first limit above exists and converges to $\int_{0}^{t}[M'(s,0)+\frac{1}{2}P'(s,0)\mbox{sgn}(\Delta Z^1_{s})]d[Z^1,Z^1]_{s}$.   

The Lebesgue-Stieltjes  integral is well-defined when the integrator is of bounded variation on $[0,t]$ and when the integrand is a bounded measurable function. Therefore, $[Z^1,Z^1]_{s}$ denoting the quadratic variation of $Z^1$ for every $s\in[0,t]$ is positive, increasing monotone and thus belongs to the bounded variation functions. Thus, when $M'(\cdot,0)+\frac{1}{2}P'(\cdot,0)\mbox{sgn}(\Delta Z^1_{\cdot})$ is bounded on $[0,t]$, and it is since $M'(\cdot,0)$ and $P'(\cdot,0)$ are c\'adl\'ag locally bounded by assumption, one may define a positive measure $\mu$ on $[0,t]$ as $\mu((s,h])=[Z^1,Z^1]_{h}-[Z^1,Z^1]_{s}$ for every $s,h\in[0,t]$, with $\mu(\{0\})=0$, such that 
 \begin{equation}
\int_{0}^{t}[M'(s,0)+\frac{1}{2}P'(s,0)\mbox{sgn}(\Delta Z^1_{s})]d\mu(s)=\int_{0}^{t}[M'(s,0)+\frac{1}{2}P'(s,0)\mbox{sgn}(\Delta Z^1_{s})]d[Z^1,Z^1]_{s}
\end{equation}
exists as a Lebesgue-Stieltjes integral. Therefore, the convergence
\begin{equation}
\lim_{n\to\infty}  \sum_{i\geq 1}[M'(\tau_{i}^{n},0)+\frac{1}{2}P'(\tau_{i}^{n},0)\mbox{sgn}(\Delta Z^1_{\tau_{i}^{n}})] (Z^1_{\tau_{i}^{n}}-Z^1_{\tau_{i-1}^{n}})^{2}
\end{equation}
can be seen as the convergence of
\begin{eqnarray}
&&\sum_{i\geq 1} [M'(\tau_{i}^{n},0)+\frac{1}{2}P'(\tau_{i}^{n},0)\mbox{sgn}(\Delta Z^1_{\tau_{i}^{n}})](Z^1_{\tau_{i}^{n}}-Z^1_{\tau_{i-1}^{n}})^{2}\nonumber\\&=&\int_{0}^{t}[M'(s,0)+\frac{1}{2}P'(s,0)\mbox{sgn}(\Delta Z^1_{s})]d\mu_{n}(s)
\end{eqnarray}
to 
\begin{eqnarray}
&&\int_{0}^{t}[M'(s,0)+\frac{1}{2}P'(s,0)\mbox{sgn}(\Delta Z^1_{s})]d\mu(s)\nonumber\\&=&\int_{0}^{t}[M'(s,0)+\frac{1}{2}P'(s,0)\mbox{sgn}(\Delta Z^1_{s})]d[Z^1,Z^1]_{s}
\end{eqnarray}
Note that we are defining $\mu_{n}$  as 
\begin{equation}
\mu_{n}((s,h])=\sum_{\substack{\tau_{i}^{n}\in\sigma_{n},\\  \tau_{i}^{n}\leq h}}(\Delta Z^1_{\tau_{i}^{n}})^{2}\delta_{\tau_{i-1}^{n}}-\sum_{\substack{\tau_{i}^{n}\in\sigma_{n},\\  \tau_{i}^{n}\leq s}}(\Delta Z^1_{\tau_{i}^{n}})^{2}\delta_{\tau_{i-1}^{n}}
\end{equation}
and $\delta_{\tau_{i-1}^{n}}$ gives the Dirac measure with mass on $t=\tau_{i-1}^{n}$. Note also that we are assuming that $Z^1_{0-}=Z^1_{0}$.

The c\'agl\'ad case works in the same manner, but now we have to deal with intervals of the form $[s,h)$. We have thus shown the following theorem.
\begin{lemma}\label{thm}
Let $M'(\cdot,0)$ and $P'(\cdot,0)$ be c\'adl\'ag locally bounded stochastic processes and let $Z^1$ be a predictable stochastic process of quadratic bounded variation as in Definition \ref{tradin}. Then  the limit as $n\rightarrow\infty$ of $\sum_{i\geq 1} [M'(\tau_{i-1}^{n},0)+\frac{1}{2}P'(\tau_{i-1},0)\mbox{sgn}(\Delta Z^1_{\tau_{i}^{n}})(\Delta Z^1_{\tau_{i}^{n}})^{2}$ exists and is given by
\begin{equation}\label{eq6}
\int_{0}^{t}[M'(s,0)+\frac{1}{2}P'(s,0)\mbox{sgn}(\Delta Z^1_{s})]d[Z^1,Z^1]_{s}
\end{equation} 
\end{lemma}
Take now the second limit and, note that 
\begin{eqnarray}
&&\lim_{n\to\infty}\sum_{i\geq1}[(M'(\tau_{i}^{n},0)+\frac{1}{2}P'(\tau_{i}^{n},0)\mbox{sgn}(\Delta Z^1_{\tau_{i}^{n}}))\nonumber\\&&-(M'(\tau_{i-1}^{n},0)+\frac{1}{2}P'(\tau_{i-1}^{n},0)\mbox{sgn}(\Delta Z^1_{\tau_{i}^{n}}))](\Delta Z^1_{\tau_{i}})^{2}\nonumber\\&&\leq\lim_{n\to\infty} \sum_{i\geq 1}2K(\Delta Z^1_{\tau_{i}^{n}})^{2}\leq 2K[Z^1,Z^1]_{t}<\infty
\end{eqnarray} 
where $K=\sup_{s\in[0,t]}|M'(s,0)+\frac{1}{2}P'(\tau_{i}^{n},0)\mbox{sgn}(\Delta Z^1_{\tau_{i}^{n}})|$. It is then easy to check that the second limit is given by  
\begin{eqnarray}\label{neweq2}
&&\sum_{s<0\leq t}[(M'(s,0)+\frac{1}{2}P'(s,0)\mbox{sgn}(Z^1_{s}-Z^1_{s-}))\nonumber\\&&-(M'(s-,0)+\frac{1}{2}P'(s-,0)\mbox{sgn}(Z^1_{s}-Z^1_{s-}))](Z^1_{s}-Z^1_{s-})^{2}
\end{eqnarray}
Since $J^{Z^1}_{\epsilon}$ is finite, the third limit also converges to
\begin{equation}\label{neweq3}
\sum_{s\in J^{Z^1}_{\epsilon} }[M'(s,0)+\frac{1}{2}P'(s,0)\mbox{sgn}(Z^1_{s}-Z^1_{s-})](Z^1_{s}-Z^1_{s-})^{2}
\end{equation} 
Assume for the moment that $M''(t,y)$ and $P''(t,y)$ are uniformly bounded in $t$ and $y$ by $L_{1}$ and $L_{2}$ so that their sums is bounded by a constant $L$. Therefore, following the proof of Theorem A.3 in \c{C}etin \textit{et al.} (2004),  the last limit satisfies the following inequality
\begin{eqnarray}
&&\lim_{n\to\infty}|\sum_{|Z^1_{\tau_{i}^{n}} - Z^1_{\tau_{i-1}^{n}}|\leq \epsilon}[o^{M}(\tau_{i}^{n},\Delta Z^1_{\tau_{i}^{n}})+\frac{1}{2}o^{P}(\tau_{i}^{n},\Delta Z^1_{\tau_{i}^{n}})\mbox{sgn}(\Delta Z^1_{\tau_{i}^{n}})]\Delta Z^1_{\tau_{i}^{n}}|\nonumber\\&&\leq L\lim_{n\to\infty}\sup_{|Z^1_{\tau_{i}^{n}} - Z^1_{\tau_{i-1}^{n}}|\leq \epsilon}|Z^1_{\tau_{i}^{n}} - Z^1_{\tau_{i-1}^{n}}|\sum_{|Z^1_{\tau_{i}^{n}} - Z^1_{\tau_{i-1}^{n}}|\leq \epsilon}(\Delta Z^1_{\tau_{i}^{n}})^{2}\leq L\epsilon[Z^1,Z^1]_{t}\nonumber\\
\end{eqnarray}
The quantity $\epsilon$ is arbitrary and therefore the last limit converges to zero as $\epsilon\rightarrow0$. It follows then that $\lim_{n\to\infty}\sum_{i\geq1}[W(\tau_{i}^{n},\Delta Z^1_{\tau_{i}^{n}})-W(\tau_{i}^{n},0)]\Delta Z^1_{\tau_{i}^{n}}$ converges to  
\begin{eqnarray}\label{m}
&&\sum_{0<s\leq t}[(M(s,Z^1_{s}-Z^1_{s-})+\frac{1}{2}P(s,Z^1_{s}-Z^1_{s-})\mbox{sgn}(Z^1_{s}-Z^1_{s-}))\nonumber\\&&-(M(s,0)+\frac{1}{2}P(s,0)\mbox{sgn}(Z^1_{s}-Z^1_{s-}))](Z^1_{s}-Z^1_{s-})\nonumber\\&&+\int_{0}^{t}[M'(s,0)+\frac{1}{2}P'(s,0)\mbox{sgn}(\Delta Z^1_{s})]d[Z^1,Z^1]_{s}\nonumber\\&&+\sum_{s<0\leq t}[(M'(s,0)+\frac{1}{2}P('s,0)\mbox{sgn}(Z^1_{s}-Z^1_{s-}))\nonumber\\&&-(M'(s-,0)+\frac{1}{2}P'(s-,0)\mbox{sgn}(Z^1_{s}-Z^1_{s-}))](Z^1_{s}-Z^1_{s-})^{2}\nonumber\\&&-\sum_{0<s\leq t}[M'(s,0)+\frac{1}{2}P('s,0)\mbox{sgn}(Z^1_{s}-Z^1_{s-})](Z^1_{s}-Z^1_{s-})^{2}\nonumber\\&&=\int_{0}^{t}[M'(s,0)+\frac{1}{2}P'(s,0)\mbox{sgn}(\Delta Z^1_{s})]d[Z^1,Z^1]^{c}_{s}\nonumber\\&&+\sum_{0<s\leq t}[(M(s,Z^1_{s}-Z^1_{s-})+\frac{1}{2}P(s,Z^1_{s}-Z^1_{s-})\mbox{sgn}(Z^1_{s}-Z^1_{s-}))\nonumber\\&&-(M(s,0)+\frac{1}{2}P(s,0)\mbox{sgn}(Z^1_{s}-Z^1_{s-}))](Z^1_{s}-Z^1_{s-})
\end{eqnarray}
as $\epsilon\rightarrow0$. 

Note that Eq. (\ref{neweq}) converges to a finite non-negative quantity  since by assumption $M(t,y)+\frac{1}{2}P(t,y)\mbox{sgn}(y)$ is non-decreasing in $y$ and the number of jumps of $Z^1$ are finite. It then follows that Eq. (\ref{m}) is a finite quantity.  

We can do the same thing with $Z^1$ being a c\'agl\'ad process. Using the same argument as above, $\lim_{n\to\infty}\sum_{i\geq1}[W(\tau_{i}^{n},\Delta Z^1_{\tau_{i}^{n}})-W(\tau_{i}^{n},0)]\Delta Z^1_{\tau_{i}^{n}}$ converges to
\begin{eqnarray}\label{m1}
&&\int_{0}^{t}[M'(s,0)+\frac{1}{2}P'(s,0)\mbox{sgn}(\Delta Z^1_{s})]d[Z^1,Z^1]^{c}_{s}\nonumber\\&&+\sum_{0\leq s< t}[(M(s,Z^1_{s+}-Z^1_{s})+\frac{1}{2}P(s,Z^1_{s+}-Z^1_{s})\mbox{sgn}(Z^1_{s+}-Z^1_{s}))\nonumber\\&&-(M(s,0)+\frac{1}{2}P(s,0)\mbox{sgn}(Z^1_{s+}-Z^1_{s}))](Z^1_{s+}-Z^1_{s})
\end{eqnarray}
For the general case, it is sufficient to define $V_{l_{1}}^{y}=\mbox{inf}\{t>0: |M''(t,y)|>l_{1}\}$, $\tilde{M}(t,y)=M''(t,y)\textbf{1}_{[0,V_{l_{1}}^{y})}$ and, $V_{l_{2}}^{y}=\mbox{inf}\{t>0: |P''(t,y)|>l_{2}\}$, $\tilde{P}(t,y)=P''(t,y)\textbf{1}_{[0,V_{l_{2}}^{y})}$. Therefore, the previous results hold for $\tilde{M}(t,y)$ and $\tilde{P}(t,y)$ for every $l_1$ and $l_2$. To prove the general result it is thus sufficient to consider processes taking values in $[-l_{1}, l_{1}]$ and $[-l_{2}, l_{2}]$. The proof is similar to the proof of It\'o Lemma in Protter (2004). 

What's important to note here is that Eqs. (\ref{m}) and (\ref{m1}) are finite. This follows from the properties of $M$ and $P$, and most importantly from the properties of $Z^1$.

Finally, the last terms of Eq. (\ref{port}) are equal to $-\frac{1}{2}P(t,0)\mbox{sgn}(Z^1_{t}-Z^1_{t-})Z^1_{t}$ for the c\'adl\'ag case and $0$ for the c\'agl\'ad case. 

We have thus proven the following theorem.
\begin{theorem}
For $Z^1$ c\'adl\'ag and $Z^1$ c\'agl\'ad, the value of the self-financing portfolio assumes the following form 
\begin{eqnarray}
&&V_{t}^{Z^1}=V_{0}^{Z^1}+\int_{0}^{t}Z^1_{s}dM(s,0)+\frac{1}{2}\int_{0}^{t}Z^1_{s}d\tilde{P}^{Z^1}(s,0) \nonumber\\&&-\int_{0}^{t}M'(s,0)d[Z^1,Z^1]^{c}_{s}-\sum_{0<s\leq t}[M(s,Z^1_{s}-Z^1_{s-})-M(s,0)](Z^1_{s}-Z^1_{s-}) \nonumber\\&&-\frac{1}{2}\int_{0}^{t}P'(s,0)\mbox{sgn}(\Delta Z^1_{s})d[Z^1,Z^1]^{c}_{s}\nonumber\\&&-\frac{1}{2}\sum_{0 <s \leq t}[P(s,Z^1_{s}-Z^1_{s-})-P(s,0)]\mbox{sgn}(Z^1_{s}-Z^1_{s-})(Z^1_{s}-Z^1_{s-})\nonumber\\&&-\frac{1}{2}P(t,0)\mbox{sgn}(Z^1_{t}-Z^1_{t-})Z^1_{t}
\end{eqnarray}
and
\begin{eqnarray}
&&V_{t}^{Z^1}=V_{0}^{Z^1}+\int_{0}^{t}Z^1_{s}dM(s,0)+\frac{1}{2}\int_{0}^{t}Z^1_{s}d\tilde{P}^{Z^1}(s,0) -\int_{0}^{t}M'(s,0)d[Z^1,Z^1]^{c}_{s}\nonumber\\&&-\sum_{0\leq s< t}[M(s,Z^1_{s+}-Z^1_{s})-M(s,0)](Z^1_{s+}-Z^1_{s}) \nonumber\\&&-\frac{1}{2}\int_{0}^{t}P'(s,0)\mbox{sgn}(\Delta Z^1_{s})d[Z^1,Z^1]^{c}_{s}\nonumber\\&&-\frac{1}{2}\sum_{0\leq s< t}[P(s,Z^1_{s+}-Z^1_{s})-P(s,0)]\mbox{sgn}(Z^1_{s+}-Z^1_{s})(Z^1_{s+}-Z^1_{s})
\end{eqnarray}
\end{theorem} 
\section{Convergence properties of the portfolio's value}\label{app2}
The following section establishes some convergence results for the portfolio's value. In particular, we show the a.s. convergence of $T_{T}^{Z^1_{n}}$  to $T_{T}^{Z^1}$, for any $Z^1_{n}$, $Z^1$ c\'agl\'ad adapted processes satisfying Definition \ref{tradin}, as long as $n\rightarrow\infty$ and $Z^1_{n}$ tends a.s. to $Z^1$,  and $[Z^1_{n},Z^1_{n}]^{c}_{s}$ a.s. to $[Z^1,Z^1]^{c}_{s}$. 

The following lemma holds for general trading strategies satisfying Definition \ref{tradin}.
\begin{lemma}\label{iku}
Let $M'(\cdot,0)$ satisfy Assumption \ref{mid}, and $Z^1_{n}$, $n\geq1$, $Z^1$ predictable processes satisfying Definition \ref{tradin}. Suppose further that $\sup_{n\geq1}[Z^1_{n},Z^1_{n}]^{c}_{T}<\infty$. If $[Z^1_{n},Z^1_{n}]^{c}_{s}$ tends a.s. to  $[Z^1,Z^1]^{c}_{s}$ for all $s\in[0,T]$, then we obtain
\begin{equation}
\lim_{n\rightarrow\infty}\int_{0}^{T}M'(s,0)d[Z^1_{n},Z^1_{n}]^{c}_{s}=\int_{0}^{T}M'(s,0)d[Z^1,Z^1]^{c}_{s}\quad \text{a.s.}
\end{equation} 
\end{lemma}
\begin{proof}
Let $Z^1_{n}, Z^1$ satisfy Definition \ref{tradin}. For $\epsilon>0$, let $s_0^{\epsilon}=0$ and  define the sequence of stopping times $(s_{m-1}^{\epsilon})_{m\geq1}$ as $s_{m}^{\epsilon}=\inf\{t>s_{m-1}^{\epsilon}:|M'(t,0)-M'(s_{m-1}^{\epsilon},0)|>\epsilon\}$. Set $M'^{\epsilon}$ equal to $M'^{\epsilon}=\sum_{i=1}^{\infty}M'(s_{i-1}^{\epsilon},0)\textbf{1}_{[\![ s_{i-1}^{\epsilon}, s_{i}^{\epsilon}[\![}$. Note that $|M'(t,0)-M'^{\epsilon}(t)|\leq\epsilon$ and $|M'^{\epsilon}(t)-M'(t,0)|\leq\epsilon$ for every $t\in[0,T]$. By construction $M'^{\epsilon}$ is c\'adl\'ag, and piecewise constant.
This means that the Lebesgue-Stieltjes integrals of $M'^{\epsilon}$ with respect to $[Z^1_{n},Z^1_{n}]$ and $[Z^1,Z^1]$ exist, and we can thus write down the following inequality
\begin{eqnarray}
&&|\int_{0}^{T}M'(s,0)d[Z^1_{n},Z^1_{n}]^{c}_{s}-\int_{0}^{T}M'(s,0)d[Z^1,Z^1]^{c}_{s}|\nonumber\\&&\leq |\int_{0}^{T}(M'(s,0)-M'^{\epsilon}(s))d[Z^1_{n},Z^1_{n}]^{c}_{s}|+|\int_{0}^{T}(M'^{\epsilon}(s)-M'(s,0))d[Z^1,Z^1]^{c}_{s}|\nonumber\\&&+|\int_{0}^{T}M'^{\epsilon}(s)d[Z^1,Z^1]^{c}_{s}-\int_{0}^{T}M'^{\epsilon}(s)d[Z^1_{n},Z^1_{n}]^{c}_{s}|
\end{eqnarray}
Then, the right-hand side is bounded by
\begin{eqnarray}
&&\epsilon([Z^1_{n},Z^1_{n}]^{c}_{T}+[Z^1,Z^1]^{c}_{T})+|\int_{0}^{T}M'^{\epsilon}(s)d[Z^1,Z^1]^{c}_{s}-\int_{(0,T]}M'^{\epsilon}(s)d[Z^1_{n},Z^1_{n}]^{c}_{s}|\nonumber\\&\leq& \epsilon(H+[Z^1,Z^1]^{c}_{T})+|\int_{0}^{T}M'^{\epsilon}(s)d[Z^1,Z^1]^{c}_{s}-\int_{0}^{T}M'^{\epsilon}(s)d[Z^1_{n},Z^1_{n}]^{c}_{s}|
\end{eqnarray}
where the last inequality follows from the fact that  $\sup_{n\geq1}[Z^1_{n},Z^1_{n}]^{c}_{T}<\infty$. 
Now, note that 
\begin{eqnarray}
&&\int_{0}^{T}M'^{\epsilon}(s)d[Z^1,Z^1]^{c}_{s}=\sum_{i=1}^{\infty} M'(s_{i-1}^{\epsilon}, 0)([Z^1,Z^1]_{s_{i}^{\epsilon}}-[Z^1,Z^1]^{c}_{s_{i-1}^{\epsilon}})\\&&\int_{0}^{T}M'^{\epsilon}(s)d[Z^1_{n},Z^1_{n}]^{c}_{s}=\sum_{i=1}^{\infty} M'(s_{i-1}^{\epsilon}, 0)([Z^1_{n},Z^1_{n}]^{c}_{s_{i}^{\epsilon}}-[Z^1_{n},Z^1_{n}]^{c}_{s_{i-1}^{\epsilon}})
\end{eqnarray}
Thus, using the convergence of $[Z^1_{n},Z^1_{n}]^{c}$ 
\begin{equation}
|\int_{0}^{T}M'^{\epsilon}(s)d[Z^1,Z^1]_{s}-\int_{0}^{T}M'^{\epsilon}(s)d[Z^1_{n},Z^1_{n}]_{s}|
\end{equation}
tends to zero as $n\rightarrow\infty$. 

The result then follows easily since $\epsilon$ was arbitrary.  
\end{proof}

Let's now show the a.s. convergence of the term 
\begin{eqnarray}
B_{T}^{Z^1_{n}}&=&\sum_{0\leq s< T}[(M(s,Z^1_{n,s+}-Z^1_{n,s})+\frac{1}{2}P(s,Z^1_{n,s+}-Z^1_{n,s})\mbox{sgn}(Z^1_{n,s+}-Z^1_{n,s}))\nonumber\\&-&(M(s,0)+\frac{1}{2}P(s,0)\mbox{sgn}(Z^1_{n,s+}-Z^1_{n,s}))](Z^1_{n,s+}-Z^1_{n,s})
\end{eqnarray}
to 
\begin{eqnarray}\label{mbaroi}
B_{T}^{Z^1}&=&\sum_{0\leq s< T}[(M(s,Z^1_{s+}-Z^1_{s})+\frac{1}{2}P(s,Z^1_{s+}-Z^1_{s})\mbox{sgn}(Z^1_{s+}-Z^1_{s}))\nonumber\\&-&(M(s,0)+\frac{1}{2}P(s,0)\mbox{sgn}(Z^1_{s+}-Z^1_{s}))](Z^1_{s+}-Z^1_{s})
\end{eqnarray}
a.s., with $Z^1_{n}$, $Z^1$ c\'agl\'ad adapted processes of bounded quadratic variation  over $[0,T]$.    

By Definition \ref{tradin}, $Z^1_{n}$ have finite number of jumps, and by Assumption \ref{assnew} $M(t,y)+\frac{1}{2}P(t,y)\mbox{sgn}(y)$ is non-decreasing in $y$. Therefore, 
\begin{eqnarray}
B_{T}^{Z^1_{n}}&=&\sum_{0\leq s< T}[(M(s,Z^1_{n,s+}-Z^1_{n,s})+\frac{1}{2}P(s,Z^1_{n,s+}-Z^1_{n,s})\mbox{sgn}(Z^1_{n,s+}-Z^1_{n,s}))\nonumber\\&-&(M(s,0)+\frac{1}{2}P(s,0)\mbox{sgn}(Z^1_{n,s+}-Z^1_{n,s}))](Z^1_{n,s+}-Z^1_{n,s})
\end{eqnarray}
converges as $n\rightarrow\infty$ to (\ref{mbaroi}).  

We have thus proved the following lemma. 
\begin{lemma}\label{iku2}
Let $M(t,y)$ and $P(t,y)$ satisfy Assumption \ref{mid}, $M(t,y)+\frac{1}{2}P(t,y)\mbox{sgn}(y)$ satisfy Assumption \ref{assnew}, and let $Z^1_{n}$ be a sequence of c\'agl\'ad processes, satisfying Definition \ref{tradin} over $[0,T]$, converging a.s. to a predictable c\'agl\'ad process $Z^1$ which also satisfies Definition \ref{tradin} over $[0,T]$. Then, we have the a.s. convergence of 
\begin{equation}
B_{T}^{Z^1_{n}}
\end{equation}
to
\begin{equation}
B_{T}^{Z^1} 
\end{equation}
\end{lemma}
\begin{remark}\label{iku3}
The convergence of the other terms of $T_{T}^{Z^1_{n}}$  can be shown in the same manner. One can simply   interchange $M'(\cdot,0)$ with $P'(\cdot,0)$ for the first lemma. 
\end{remark}
\section{Standard economy}\label{app1}
\begin{definition}\label{def5}
Let $V_{0}^{Z^1}=0$. The value of the portfolio in the standard economy is given by $\int_{0}^{T}Z^1_{s}dM(s,0)$ for predictable trading strategies $Z^1$. A portfolio is called admissible if it satisfies Equation (\ref{eqport2}), thus  trading strategies that require unlimited capital are not allowed, in particular doubling strategies have to be excluded. An admissible strategy $Z^1$ is an arbitrage on $[0, T]$ if  $\int_{0}^{T}Z^1_{s}dM(s,0) \geq 0$ a.s. $\mathbb{P}$, and $\mathbb{P} (\int_{0}^{T}Z^1_{s}dM(s,0) >0)>0$.  A FLVR is a sequence of random variables  $Z^1_{n}$ such that $\int_{0}^{T}Z^1_{n,s}dM(s,0)\geq-\frac{1}{n}$  and $V_{T}^{Z^1_{n}}$ converges a.s. to some limit $V\in[0,\infty]$ a.s., with $V$ not identically zero. The standard economy has the NFLVR property if FLVR condition fails. 
\end{definition}
\begin{theorem}\label{std}
Let $S$ be a locally bounded positive semimartingale. Then, there is NFLVR in the standard economy if and only if there exists a $\mathbb{Q}$-local martingale measure equivalent to $\mathbb{P}$ such that $S$ is a $\mathbb{Q}$-local martingale.
\end{theorem}
\begin{proof}
See Delbaen \& Schachermayer (1994) for a formal proof. 
\end{proof}
\begin{theorem}\label{sftap}
Let $M(\cdot,0)$ be as in Sec. \ref{sec2}. Suppose the equivalent $\mathbb{Q}$-local martingale measure is unique. Then the standard economy is complete.
\end{theorem}
\begin{proof}
See Harrison \& Pliska (1981) and Protter (2001) for the proof.
\end{proof}



\end{document}